\documentclass[a4paper,11pt]{article}

\usepackage{amsmath,amsfonts,amssymb,caption,graphicx}
\usepackage{ushort}
\usepackage{makeidx}
\usepackage{float}
\usepackage{accents}
\usepackage{color}
\usepackage{framed}
\usepackage{versions}
\usepackage{xcolor}
\usepackage{emptypage}
\usepackage{array}
\usepackage{multirow}
\usepackage{subcaption}

\usepackage{amsthm}

\usepackage{wrapfig}
\usepackage{booktabs}       % professional-quality tables
\def\hmath$#1${\texorpdfstring{{\rmfamily\textit{#1}}}{#1}}

\usepackage{tikz}
\usetikzlibrary{trees}

\numberwithin{equation}{section}
\theoremstyle{plain}
\newtheorem{thm}{Theorem}[section]
\newtheorem{lem}{Lemma}[section]
\newtheorem{prop}{Proposition}[section]
\newtheorem{cor}{Corollary}[section]

\usepackage[bookmarks]{hyperref}
\hypersetup{
    colorlinks=true,   	% false: boxed links; true: colored links
    linkcolor=red,      % color of internal links (change box color 
    citecolor = [rgb]{0 0.7 0},   	% color of links to bibliography
    filecolor=magenta, 	% color of file links
    urlcolor=blue,
    linktocpage=true
}
 
\input{epsf}

\def\ba{\begin{align}}
\def\ea{\end{align}}
\def\ban{\begin{align*}}
\def\ean{\end{align*}}

\def\be{\begin{eqnarray}}
\def\ee{\end{eqnarray}}
\def\ben{\begin{eqnarray*}}
\def\een{\end{eqnarray*}}

\def\bqq{\begin{equation}}
\def\eqq{\end{equation}}
\def\bqqn{\begin{equation*}}
\def\eqqn{\end{equation*}}

\def\elabel#1{\label{e:#1}}

\def\sq{$\Box$}

\def\qed{\ifmmode\sq\else{\unskip\nobreak\hfil
\penalty50\hskip1em\null\nobreak\hfil\sq
\parfillskip=0pt\finalhyphendemerits=0\endgraf}\fi\par\medbreak}

\newsavebox{\junk}
\savebox{\junk}[1.6mm]{\hbox{$|\!|\!|$}}

\def\til={{\widetilde =}}

 \def\eq#1/{(\ref{#1})}

\def\eq#1/{(\ref{e:#1})}

\newcommand{\beqn}[1]{\notes{#1}%
\begin{eqnarray} \elabel{#1}}

\newcommand{\eeqn}{\end{eqnarray} } 

\newcommand{\beq}[1]{\notes{#1}%
\begin{equation}\elabel{#1}}

\newcommand{\eeq}{\end{equation}} 

\def\bdes{\begin{description}}
\def\edes{\end{description}}

\def\notes#1{}

\definecolor{mag}{rgb}{0.7,0,0.3}
\definecolor{dgreen}{rgb}{0.1,0.5,0.1}
\definecolor{dred}{rgb}{.8,0,0}
\definecolor{gray}{rgb}{.8,.8,.8}
\definecolor{brown}{rgb}{0.6451,0.3706,0.1745}

\usepackage{natbib}

\begin{document}

\title{Posterior Representations for Bayesian Context Trees: \\ Sampling, Estimation and Convergence}

\author
{

 	Ioannis Papageorgiou
    \thanks{Department of Engineering,
        University of Cambridge,
        Trumpington Street, Cambridge CB2 1PZ, UK.
                Email: \texttt{\href{mailto:ip307@cam.ac.uk}%
			{ip307@cam.ac.uk}}.
 	}
 \and
    Ioannis Kontoyiannis
    \thanks{Statistical Laboratory, Centre for Mathematical Sciences, University of Cambridge, Wilberforce Road, Cambridge CB3 0WB, UK.
                Email: \texttt{\href{mailto: yiannis@maths.cam.ac.uk}%
			{ yiannis@maths.cam.ac.uk}}.
	I.K. was 
supported in part by the Hellenic
Foundation for Research and Innovation (H.F.R.I.) under the ``First Call for 
H.F.R.I.\ Research Projects to support
Faculty members and Researchers and the procurement of high-cost research equipment 
grant,'' project number 1034.
        }
}

\maketitle

%\begin{center}
 %   \large \textbf{Abstract}
%\end{center}

\begin{center}
    \large \textbf{Abstract}
\end{center}

We revisit the
Bayesian Context Trees (BCT) modelling framework for 
discrete time series, which was recently found 
to be very effective in numerous tasks 
including model selection, estimation and prediction. 
A novel representation
of the induced posterior distribution on model space
is derived in terms of a simple branching process,
and several consequences of this are 
explored in theory and in practice.
First, it is shown that the branching process
representation leads to a simple variable-dimensional 
Monte Carlo sampler for the joint posterior distribution 
on models and parameters, which can efficiently produce
independent samples. This sampler is found to be
more efficient than earlier MCMC samplers for the same tasks.
Then, the branching process representation
is used to establish the asymptotic consistency 
of the BCT posterior, including the derivation 
of an almost-sure convergence rate. Finally,
an extensive study is carried out
on the performance of the induced Bayesian 
entropy estimator.
Its utility is illustrated
through both
simulation experiments and real-world applications,
where it is found to outperform several state-of-the-art
methods.  \\


%\footnotetext[{}]{geia}

{\let\thefootnote\relax\footnote{{Preliminary versions of some of the results in this work were presented in~\cite{branch_isit}} and~\cite{our_entropy}.}}

\newpage

%\tableofcontents

\section{Introduction}

The statistical modelling and analysis of discrete time series are 
important scientific and engineering tasks, with 
a very wide range of applications. Numerous
Markovian model classes have been developed in 
connection with these and related problems,
including mixture transition distribution (MTD) 
models \citep{raftery1985model,berchtold2002mixture}, variable-length 
Markov chains
(VLMC) \citep{buhlmann1999variable,buhlmann2000model,machler2004variable} 
and sparse Markov chains \citep{jaaskinen2014sparse,xiong2016recursive}. Alternative approaches also include the use of multinomial logit or probit regression \citep{zeger1986longitudinal}, 
categorical regression models \citep{fokianos:03},
and conditional tensor factorisation \citep{sarkar2016bayesian}.

A popular and useful class of relevant models 
for discrete time series are the {\em context-tree sources}, 
introduced by \cite{ris83a,ris83b,ris86} as descriptions of variable-memory Markov chains, a flexible and rich class of chains that admit parsimonious 
representations. Their key feature is that the memory length
of the chain is allowed to 
depend on the most recently observed symbols, providing
a richer model class than ordinary Markov chains. 
Context-tree sources have been very successful in
information-theoretic applications in connection with
data compression~\citep{weinb,ctw}, and the celebrated
{\em context tree weighting} (CTW) algorithm~\citep{ctw,ctw2},
also based on context-tree sources, has 
been used widely as an efficient compression method
with extensive theoretical guarantees and justifications.

Recently, \cite{our}  revisited context-tree models and 
the CTW algorithm from a Bayesian inference point of view. A general 
modelling framework, called {\em Bayesian Context Trees} (BCT), 
was developed for discrete time series, along with a collection of 
efficient algorithmic tools both for exact inference and for posterior sampling via Markov chain Monte Carlo (MCMC).
The BCT methods were found to be 
very effective in important statistical tasks, including model selection, 
estimation, prediction and change-point detection~\citep{ourisit,lungu2022bayesian,changepoint}; see also the R package BCT~\citep{bctrpackage}.

In this work we derive an alternative representation
of the posterior distribution induced by the BCT framework,
and explore several ways in which it facilitates inference,
both in theory and in practice. Our
first main contribution is described 
in Sections~\ref{s:prior} and~\ref{s:posterior},
where we show
that both the prior and posterior distributions on
model space admit explicit representations 
as simple branching processes. In particular,
sampling tree models from the prior or the
posterior is shown to be equivalent to
generating 
trees via an appropriate Galton-Watson 
process \citep{athreya2004branching,harris1963theory}, 
stopped at a given depth. 
Therefore, in
some sense the BCT model prior acts as a `conjugate' prior 
for variable-memory Markov chains.
An immediate first practical consequence of this
representation
is that it facilitates direct Monte Carlo (MC) sampling
from the posterior, where independent and identically
distributed (i.i.d.) samples can be efficiently obtained
from the joint posterior distribution on models and
parameters. This variable-dimensional sampler and its
potential utility in a wide range of applications
(including model selection, parameter estimation
and Markov order estimation)
are described in Section~\ref{s:sampling}.

The Bayesian perspective adopted in this work is neither 
purely subjective nor purely objective. For example,
we think of the model posterior distribution as a summary
of the most accurate, data-driven representation
of the regularities present in a given time series, 
but we also examine the frequentist properties of
the resulting inferential procedures
\citep{gelman1995bayesian,chipman:01,bernardo2009bayesian}.
Indeed, in 
Section~\ref{theory} we employ the 
branching process representation
to show that the posterior asymptotically
almost surely 
concentrates on the ``true'' underlying model 
(Theorem~\ref{post_conc_thm}), and in 
Theorem~\ref{t:postrate} we derive an 
explicit rate for this convergence
as a function of the sample size.
Analogous results are established in
Theorems~\ref{t:out} and \ref{t:outrate}
in the case of out-of-class modelling,
when the data are generated by a model
outside the BCT class. Importantly,
the limiting model is explicitly 
identified in this case.
The branching process representation
is also used
in Proposition~\ref{pred_distr}
to provide a simple, explicit representation
of the posterior predictive distribution.
These theoretical results are the
second main contribution of this work.

Our last contribution,
in Section~\ref{exp},
is a brief experimental
evaluation of the utility of the MC
sampler of Section~\ref{branch},
and a careful examination of the performance
of the induced Bayesian {\em entropy estimator}.
In Section~\ref{s:MCMC}, the new i.i.d.\ sampler
is compared with the MCMC samplers introduced by \cite{our}
on simulated data.
As expected,
it is found that
the i.i.d.\ sampler has superior performance,
both in terms of estimation accuracy and,
as expected, in terms of mixing.

Finally, in Section~\ref{s:entropy} we consider
the important problem of estimating the entropy rate
of a discrete time series. Starting with the 
original work of \cite{shannon1951prediction}, 
many different approaches have been developed for this task,
including
Lempel-Ziv~(LZ) estimators 
\citep{ziv1977universal,wyner1989some}, prediction by partial matching~(PPM) 
\citep{cleary1984data}, the CTW algorithm \citep{gao2008estimating},
and block sorting methods \citep{cai2004universal}; 
for an extensive review of the relevant
literature, see \cite{verdu2019empirical}. 
Entropy estimation has also received a lot 
of attention 
in the neuroscience literature
\citep{strong1998entropy,london2002information,nemenman2004entropy,paninski2003estimation},  in an effort to describe and quantify the amount of information transmitted by~neurons.

In contrast with most earlier work, here we adopt
a fully-Bayesian approach. Since the entropy
rate is a functional of the
model and associated parameters, the 
branching process sampler of Section~\ref{branch}
makes it possible to effectively sample
from (and hence estimate) the actual posterior
distribution of the entropy rate.
This of course provides a much richer picture
than the simple point estimates employed in most
applications. The performance of the BCT
entropy estimator is illustrated both on
simulated data and
real-world applications from neuroscience, 
finance and animal communication,
where it is seen to outperform several of the
state-of-the-art methods.

In closing this introduction, we mention that 
there are, of course, numerous other approaches
to the problem of inference for discrete time
series. In addition to the extensive review
given by \cite{our}, those include the class
of reversible variable-memory chains 
examined by \cite{bacallado:11,bacallado2013bayesian,bacallado:16},
and the Bayesian analyses of discrete models with
priors that encourage sparse representations
developed in \citet{heiner:19,heiner:22}.

\newpage

\section{Bayesian context trees} 
\label{bct}
%%%%%%%%%%%%%%%%%%%%%%%%%%%%%%%%%%%%%%%%%%%%%%%%%%%%%%%%%%%%%%%%%%%%

In this section, we briefly review the BCT model class,
the associated prior structure, and some relevant 
properties and results that will be needed in subsequent sections.

\smallskip

The BCT model class consists of
variable-memory Markov chains, where
the memory length of the process may
depend on the values of the 
most recently observed symbols. 
Variable-memory Markov chains 
admit natural representations as context trees. 
Let~$\{X_n\}$ 
be a $d$th order Markov chain,
for some~$d\geq 0$, taking values in the alphabet 
$A=\{0,1,\ldots,m-1\}$.
The {\em model} describing $\{X_n\}$ 
as a variable-memory chain is represented by
a proper $m$-ary tree $T$ as in the example 
in Figure~\ref{tree},
where a tree $T$ is called {\em proper} if any node in $T$ 
that is not a leaf has exactly $m$ children. 
% , which is taken from \cite{our}.

\smallskip

Each leaf of the tree $T$ corresponds to a string $s$ determined by the 
sequence of symbols along the path from the root node $\lambda$ to that leaf. 
At each leaf $s$, there is an associated set of parameters~$\theta_s$
that form a probability vector,
$\theta_s=(\theta_s(0),\theta_s(1),\ldots,\theta_{s}(m-1))$.
At every time $n$, the conditional distribution of the next symbol $X_n$, 
given the past~$d$ observations $(x_{n-1}, \ldots , x_{n-d})$, is 
given by the vector $\theta_s$ associated to the unique leaf~$s$ of $T$ 
that is a suffix of $(x_{n-1}, \ldots , x_{n-d})$. 
Throughout this paper, every variable-memory Markov chain is described by a tree model $T$ and a set of associated parameters $\theta=\{\theta_s\;;\;s\in T\}$,
where~$T$ is viewed as the collection of its~leaves.

\begin{figure}[!ht]
\centering
% \vspace*{-0.1 cm}
 \includegraphics[width= 0.6 \linewidth]{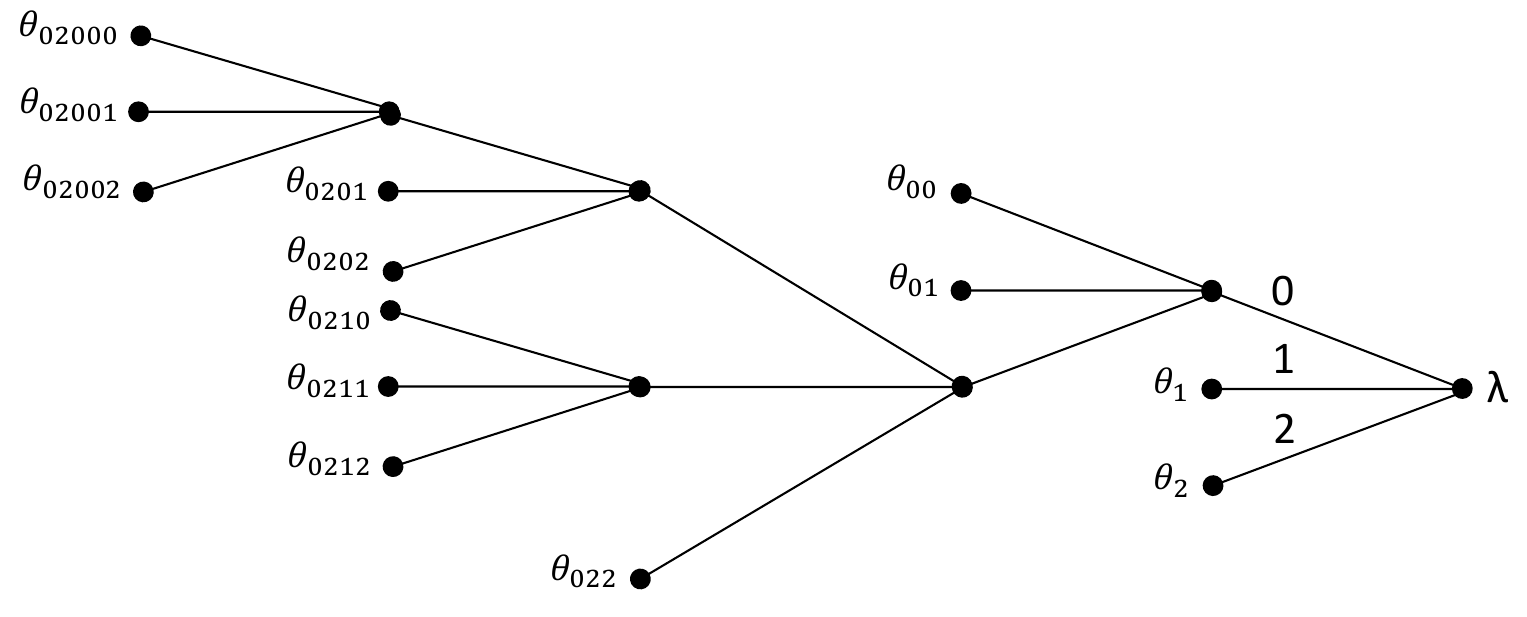}
\vspace*{-0.2 cm}
\caption{Tree model representation of a 5th order variable-memory chain.}
\label{tree}
\end{figure}

% \vspace*{-1.0 cm}

\noindent
{\bf Model prior. } Given a maximum depth $D\ge 0$, 
let $\mathcal {T} (D) $ denote the collection of all proper $m$-ary trees 
with depth no greater than $D$.
In \cite{our}, the following prior distribution is introduced 
on $\mathcal {T} (D)$,
\begin{equation}\label{prior_T}
\pi(T)=\pi_D(T;\beta)=\alpha^{|T|-1}\beta^{|T|-L_D(T)} \ ,
\end{equation}
where $\beta\in(0,1)$ is a hyperparameter, $\alpha$ is given 
by $\alpha=(1-\beta)^{1/(m-1)}$,
$|T|$ is the number of leaves of $T$,
and $L_D(T)$ is the number of leaves of 
$T$ at depth $D$.

\smallskip

This prior clearly penalises larger trees by an exponential amount,
and larger values of $\beta$ make the penalisation more severe. 
We adopt the default value
$\beta = 1 - 2 ^ {-m+1}$ for~$\beta$; see
\cite{our} for an extensive discussion of the properties
of this prior and the choice of the hyperparameter $\beta$.

 \medskip

\noindent
{\bf Prior on parameters. } Given a tree model $T \in \mathcal{T} (D)$, an independent Dirichlet prior
with parameters $(1/2,1/2,\ldots,1/2)$ is placed
on each $\theta_s$, so that: 
\begin{equation}\label{prior_theta}
\pi (\theta | T) = \prod _ {s \in T} \pi (\theta_s) = \prod _ {s \in T} \text{Dir}(1/2, 1/2, \ldots , 1/2).   \vspace*{-0.1 cm}
\end{equation}

Let $x=(x_{-D+1},\ldots,x_0,x_1,\ldots,x_n)$ denote a time series
with values in $A$. For each~$i\leq j$, we write $x_i^j$
for the segment $(x_i,x_{i+1},\ldots,x_j)$, so that
$x$ consists
of the observations $x_1^n$ along with an initial context
$x_{-D+1}^0$ of length $D$.

\smallskip

One of the main observations of \cite{our} is that
the prior predictive likelihood, averaged over
both models and parameters, \vspace*{-0.1 cm}
\begin{equation}
P(x)=\sum_{T\in\mathcal {T}(D)}
\int_\theta P(x|T,\theta)\pi(\theta|T)\pi(T)d\theta,  \vspace*{-0.1 cm}
\label{eq:ppl}
\end{equation}
can be computed exactly and efficiently by a version of the CTW
algorithm (where $P(x|T,\theta)$ denotes
the probability of $x$ under model $T$ with parameters $\theta$),
which
of course facilitates numerous
important statistical tasks.
For a given time series 
$x=x_{-D+1}^n$,
the CTW algorithm uses the \textit{estimated probabilities} $P_{e,s}$
defined as follows. For any tree model $T\in\mathcal{T}(D)$ and any
context $s$ (not necessarily a leaf),  \vspace*{-0.05 cm}
\begin{equation} 
P_{e,s} = P_{e}(a_s)= 
\frac{\prod_{j=0}^{m-1} 
[(1/2)(3/2)\cdots 
(a_s(j)- 1/2)]}
{(m/2)(m/2+1) 
\cdots (m/2+M_s-1)},
\label{eq:Pe}
\end{equation} 
where the elements of each \textit{count vector}
$a_s=(a_s(0),a_s(1),\ldots,a_s(m-1))$
are given by,
\begin{equation}
a_s(j)=\mbox{\# times symbol $j\in A$ follows 
context $s$ in $x_1^n$},
\label{eq:vector-a}
\end{equation}
and $M_s=a_s(0)+a_s(1)+\cdots+a_s(m-1)$.

% \newpage

\medskip

\noindent
{\bf CTW: The context tree weighting algorithm. }
\vspace*{-0.05in}
\begin{enumerate}
\item Build the tree $T_{\text{MAX}}$, 
which is the smallest proper tree that contains
all the contexts $x_{i-D+1} ^ {i}, \  i=1,2,\ldots,n$,
as leaves. Compute $P_{e,s}$ as given in (\ref{eq:Pe}) for 
each node $s$ of $T_{\text{MAX}}$.

\item Starting at the leaves and proceeding recursively towards the root, for each node~$s$ of $T_{\text{MAX}}$ compute the \textit{weighted probabilities} $P_{w,s}$, given by, \begin{equation} \label{pws}
P_{w,s}\!=\!
\left\{
\begin{array}{ll}
P_{e,s},  \; \; \; &\mbox{if $s$ is a leaf,}\\
\beta  P_{e,s}+(1-\beta)  \prod_{j=0}^{m-1} P_{w,sj},  \; \; \; &\mbox{otherwise,}
\end{array}\!\!
\right.\!\!
\end{equation} where $sj$ is the concatenation
of context $s$ and symbol $j$.
\end{enumerate}

%\vspace*{-0.1in}

As shown in \cite{our}, the weighted probability 
$P_{w,\lambda}$ produced by the CTW algorithm at the root 
$\lambda$, is indeed exactly equal to the prior predictive 
likelihood in~(\ref{eq:ppl}).
A family of Markov chain Monte Carlo~(MCMC) samplers 
for the posterior $\pi(T|x)$ or $\pi(T,\theta|x)$
were also introduced in \cite{our}. 
However, as will be seen below,
the representation of Section~\ref{branch}
leads to a simple i.i.d.\ MC sampler
that typically outperforms these MCMC samplers.

%\section{An equivalent branching process description} \label{branch} 

\section{Branching process representations} \label{branch}

In this section we show that 
the BCT model prior $\pi(T)=\pi_D(T;\beta)$ 
of~(\ref{prior_T}) and
the resulting posterior $\pi(T|x)$ 
both admit natural and easily interpretable representations
in terms of simple branching processes.
We also discuss how the posterior representation
leads to efficient samplers that can be used 
for model selection and estimation.

\subsection{The prior branching process}
\label{s:prior}

Given $D\geq 0$ and $\beta\in(0,1)$,
let $T=\{\lambda\}$ consist of only the root
node $\lambda$ and consider the following procedure:
\begin{itemize}
\item
If $D=0$, stop.
\item
If $D>0$, then, with probability $\beta$, mark the root as a 
leaf and stop, or, with probability $(1-\beta)$,
add all~$m$ children of $\lambda$ at depth~$1$ to~$T$. 
If \mbox{$D=1$}, stop.
\item
If $D>1$, 
examine each of
the $m$ new nodes, and either
mark a node as a leaf with probability $\beta$,
or add all~$m$ of its children to $T$ with 
probability $(1-\beta)$, independently from node to node. 
\item
Continue recursively, at each
step examining all non-leaf nodes at depths~strictly
smaller than $D$, until no more eligible nodes remain
to be examined.
\item
Output the resulting tree $T \in \mathcal {T} (D)$.
\end{itemize}

The above construction is a simple
Galton-Watson process \citep{athreya2004branching,harris1963theory}
with offspring distribution $(\beta,(1-\beta))$
on $\{0,m\}$, stopped at generation~$D$. The following proposition 
states that
the distribution of a tree $T$ generated by this process is
exactly the prior $\pi_D(T;\beta)$. Note that this also 
implies that the expression for $\pi_D(T;\beta)$ given
in~(\ref{prior_T}) indeed defines
a probability distribution on $\mathcal {T} (D)$, 
giving an alternative proof of \cite[Lemma~2.1]{our}.

\begin{prop}
\label{lem:branch_prior}
For any $D\geq 0$ and any $\beta\in(0,1)$,
the probability that the above
branching process produces any particular
tree $T\in\mathcal {T} (D)$ is given by $\pi_D(T;\beta)$
as in~{\em (\ref{prior_T})}.
\end{prop}

\begin{proof}
When $D=0$, $\mathcal {T} (D)$ consists of a single tree,
$T=\{\lambda\}$, which has probability~1 under 
both $\pi_D$ and the branching 
process construction. 
Assume $D\geq 1$. Note that every tree $T\in \mathcal {T} (D) $
can be viewed as a collection of a number, $k$, say,
of $m$-branches, since every node in $T$ 
has either zero or $m$ children. The proposition is proven 
by induction on $k$.
The result is trivial for $k=0$,
since the only tree with no $m$-branches
is $T=\{\lambda\}$ and its probability
under
both $\pi_D$ and the 
branching process construction is equal to~$\beta$.

\smallskip

Suppose the claim of the proposition is true for all trees 
with $k$ $m$-branches, and let
$T'\in \mathcal {T} (D)$ consist of $(k+1)$ $m$-branches.
Then $T'$ can be obtained from some $T\in \mathcal {T} (D)$
that has $k$ $m$-branches, by adding a single $m$-branch
to one of its leaves, $s$, say.
Two cases are considered.

\smallskip

$(i)$
If $s$ is at depth $D-2$ or smaller,
then the probability $\pi_b(T')$ of $T'$ under the branching 
process construction is, \vspace*{-0.1 cm}
$$\pi_b(T')=\frac{\pi_b(T)}{\beta}(1-\beta)\beta^m
=\frac{\pi_D(T;\beta)}{\beta}(1-\beta)\beta^m
=\frac{\alpha^{|T|-1}\beta^{|T|-L_D(T)}}{\beta}\alpha^{m-1}\beta^m, 
\vspace*{-0.1cm} $$
where the second equality follows from the inductive hypothesis 
and the third from the definition of $\pi_D (T;\beta)$.
Therefore, since $|T'|=|T|+m-1$ and no leaves are added at depth $D$, 
so that $L_D(T')=L_D(T)$,
$$\pi_b(T')
=\alpha^{[|T|+m-1]-1}\beta^{|T|+m-1-L_D(T)}
=\alpha^{|T'|-1}\beta^{|T'|-L_D(T')}=\pi_D(T';\beta),$$
as required.

$(ii)$
If $s$ is at depth $D-1$, we similarly find that,  \vspace*{-0.1 cm}
$$\pi_b(T')=\frac{\pi_b(T)}{\beta}(1-\beta)
=\frac{\pi_D(T;\beta)}{\beta}(1-\beta)
=\frac{\alpha^{|T|-1}\beta^{|T|-L_D(T)}}{\beta}\alpha^{m-1}, \vspace*{-0.1 cm} $$
and since $|T'|=|T|+m-1$, but now, $L_D(T')=L_D(T)+m$,
$$\pi_b(T')
=\alpha^{[|T|+m-1]-1}\beta^{[|T|+m-1]-[L_D(T)+m]}
=\alpha^{|T'|-1}\beta^{|T'|-L_D(T')}=\pi_D(T';\beta),$$
completing the proof. 
\end{proof}

\vspace*{-0.2 cm}

Apart from being aesthetically appealing, 
this representation 
also offers a simple and practical way of 
sampling from $\pi_D(T;\beta)$. Moreover, 
using well-known properties of the Galton-Watson process
we can perform some 
direct computations that offer 
better insight into the nature and specific properties 
of the BCT prior.

\medskip

\noindent
{\bf Interpretation and choice of $\beta$. }
The branching process description of $\pi_D(T,\beta)$  further clarifies the role of the hyperparameter $\beta$: It is exactly the probability that, when a node is added to the tree $T$, it is marked as a leaf and its children are not included in $T$.
%The branching process description of $\pi_D(T;\beta)$ 
%further clarifies the way in which larger models are penalised:
%Adding $m$ children to any node decreases the model
%probability by a factor of exactly $(1-\beta)$.

In terms of choosing the value of the hyperparameter
$\beta$ appropriately, 
recall that, for a Galton-Watson process, 
the expected number of children of each node,
in this case $\rho=m(1-\beta)$, 
governs the probability of extinction $P_e$:
If $\rho \leq 1$ we have $P_e=1$, whereas if
$\rho>1$, $P_e$ is strictly less than one.
Therefore, in the 
binary case $m=2$, the original choice $\beta=1/2$ 
used in the CTW algorithm
gives an expected number of children equal to the critical 
value $\rho=1$. This suggests that a reasonable choice for 
general alphabets could be $\beta = 1 -1/m$, which keeps $\rho =1$, 
so that the resulting prior would have similar 
qualitative characteristics with the well-studied binary case. 
This is also in line with the observation of \cite{our} that 
$\beta$ should decrease with $m$.

Now suppose $T$ is a random model generated by the prior
and let $L_d(T)$ denote the number of nodes at depth $d=0,1,\ldots,D$.
Then, standard Galton-Watson theory 
\citep{athreya2004branching,harris1963theory} provides
the useful expressions, 
\[
\mathbb {E}\left  [ L_d (T) \right ] = \rho ^ d    , \quad \quad \quad  \text{Var} \left  [ L_d (T) \right ]  = \left\{
\begin{array}{ll}
 \sigma ^2 d,  \; \; \; &\mbox{if $\rho =1$,}\\
\sigma ^2 \rho  ^ {d-1} \frac{1 - \rho ^ d}{1 - \rho},  \; \; \; &\mbox{if $\rho \neq 1$},
\end{array}\!\!
\right. \vspace*{-0.15 cm}
\]
where $\sigma ^ 2 = m ^ 2 \beta (1 - \beta ) $.

\subsection{The posterior branching process}
\label{s:posterior}

Given a time series $x=x_{-D+1}^n$, a maximum depth $D$,
and $\beta\in(0,1)$,
for any context $s$ with length strictly smaller than $D$
we define the \textit{branching probabilities} $P_{b,s}$ as, \vspace*{-0.1 cm}
\begin{equation} \label{pbs}
P_{b,s}:=\frac{\beta P_{e,s}}{P_{w,s}}, \vspace*{-0.1 cm}
\end{equation}
where the estimated and weighted probabilities, $P_{e,s}$ 
and $P_{w,s}$, are defined in~(\ref{eq:Pe})-(\ref{pws}), 
and with the convention that $P_{b,s} = \beta$ for all contexts $s$ 
that do not appear in $x$. Starting with $T=\{\lambda\}$,
the following construction
produces a sample model $T\in{\mathcal T}(D)$ from the model posterior $\pi(T|x)$:

\begin{itemize}
\item
If $D=0$, stop.
\item
If $D>0$, then, with probability $P_{b,\lambda}$, 
mark the root as a leaf and stop,
or, with probability $(1-P_{b,\lambda})$,
add all~$m$ children of $\lambda$ at depth $1$ to $T$. 
If $D=1$, stop.
\item
If $D>1$,
examine each of
the $m$ new nodes and either
mark a node $s$ as a leaf with probability $P_{b,s}$,
or add all~$m$ of its children to $T$ with probability
$(1-P_{b,s})$, independently from node to node. 
\item
Continue recursively, at each
step examining all non-leaf nodes at depths strictly
smaller than $D$, until no more eligible nodes remain.
\item
Output the resulting tree $T \in \mathcal {T} (D)$.
\end{itemize}

\begin{prop}
\label{p:branching_post}
For any $D\geq 0$ and any $\beta\in(0,1)$,
the probability that the above
branching process produces any particular
tree $T\in\mathcal {T}(D)$ is given by $\pi(T|x)$.
\end{prop}

\vspace*{-0.2 cm}

The proof, which follows along the same lines as that of 
Proposition~\ref{lem:branch_prior}, is given in Section~A of the
supplementary material. 
It is perhaps somewhat remarkable that the posterior 
$\pi(T|x)$ on 
the vast model space~$\mathcal {T}(D)$ admits such a simple 
description. Indeed, the posterior branching process is of
exactly the same form as that of the prior,
which can then
naturally be viewed as a conjugate prior on $\mathcal {T}(D)$. 

\medskip

\noindent
{\bf Model posterior probabilities. } 
Proposition~\ref{p:branching_post}
allows us to write an exact expression for the posterior 
of any model $T \in \mathcal {T} (D)$ in terms of
the branching probabilities $P_{b,s}$, \vspace*{-0.1 cm}
\begin{equation} \label {post_alt_expr}
\pi (T|x ) = \prod _ {s \in T _{o}} \left ( 1 - P_{b,s}\right ) \prod _ {s \in T} P_{b,s}, \vspace*{-0.1 cm}
\end{equation}
where $T_{o}$ denotes the set of all \textit{internal} nodes of $T$, 
and with the convention that $P_{b,s}=1$ for all leaves of $T$ 
at depth $d=D$. This expression will be the starting
point in the proofs of the
asymptotic results of Section~\ref{theory} for $\pi(T|x)$.

% \newpage

In terms of inference,
the main utility of Proposition~\ref{p:branching_post}
is that 
it offers a practical way of obtaining exact i.i.d.\ samples 
directly from the model posterior, as described in the next section.

\subsection{Sampling from the posterior} \label{section3.3}
\label{s:sampling}

The branching process representation of 
$\pi(T|x)$ readily leads to a simple way for obtaining
i.i.d.\ 
samples $\{T^{(i)}\}$ from
the posterior on model space.
And since the
full conditional density of 
the parameters $\pi(\theta |T,x)$ is explicitly identified
by \cite{our} as a product of Dirichlet densities,
\begin{equation} \label{post_theta}
\pi(\theta |T,x) = \prod _ {s \in T} \text{Dir}
\left (1/2 + a_s (0), 1/2 + a_s (1), \ldots,1/2 + a_s (m-1) \right ),
\end{equation}
for each $T^{(i)}$
we can draw a conditionally independent sample 
$ \theta ^ {(i)} \sim \pi ( \theta | T ^ {(i)},x  )$,
producing a sequence of
exact i.i.d.\ samples $\{(T ^ {(i)}, \theta ^ {(i)})\}$ from 
the joint posterior $\pi(T,\theta|x)$. 

This facilitates numerous applications.
For example, effective parameter estimation can be 
performed by simply keeping the samples 
$\{ \theta ^ {(i)}\}$, which come from the marginal posterior 
distribution $\pi(\theta |x)$. Similarly, Markov order
estimation can be performed by collecting the sequence
of maximum depths of the models $\{T^{(i)}\}$.
And in model selection tasks, the model posterior can 
be extensively explored, offering better insight and 
deeper understanding of the underlying structure 
and dependencies present in the data. 

Although a family of MCMC samplers was introduced 
and successfully used for the same
tasks in \cite{our}, MCMC sampling has well-known limitations
and drawbacks, including potentially slow mixing, 
high correlation between samples, and the need for
convergence diagnostics
\citep{gelman1992inference,cowles1996markov,robert2004monte}.
Partly for these reasons, 
being able to obtain i.i.d.\ samples from the posterior 
is generally much more desirable, as illustrated 
in Section~\ref{s:MCMC}.

\medskip

\noindent{\bf Estimation of general functionals. }
Consider the general Bayesian estimation problem, where 
the goal is to estimate an arbitrary functional $F=F (T,\theta)$ 
of the underlying variable-memory chain, based on data $x$. 
Using the above sampler, the entire posterior distribution
of the statistic $F$ can be explored, by considering
the i.i.d.\ samples 
$F ^ {(i)} = F ( T ^ {(i)} , \theta ^ {(i)} )$,
distributed according to the desired posterior $\pi(F|x)$.

In connection with classical estimation 
techniques, and in order to evaluate estimation 
performance more easily in practice, several reasonable 
point estimates can also be obtained.
The most common choices are the empirical average
approximation to the posterior mean,
% \vspace*{-0.1in}
\begin{equation} \label {post_mean_mc}
\widehat F _ {\text{MC}} = \frac{1}{N} \sum _{i=1} ^ N F^{(i)} = 
\frac{1}{N} \sum _{i=1} ^ N  F ( T ^ {(i)}, \theta ^ {(i)}),
% \vspace*{-0.1in}
\end{equation}
or the posterior mode, i.e., the maximum 
\textit{a posteriori} probability (MAP) estimate,
$\widehat F _ {\text {MAP}} $. 
In cases where the conditional mean 
$\bar{F}(T)= {E} \big ( F(T,\theta)|x,T \big )$ can be computed for any 
model~$T$ (as e.g. in the case of parameter estimation), a 
lower-variance Rao-Blackwellised
estimate \citep{blackwell1947conditional,gelfand1990sampling} 
for the posterior mean can also be obtained as,
$$\widehat F _ {\text{RB}} = \frac{1}{N} \sum _{i=1} ^ N \bar F(T ^{(i)}) .$$
Importantly, as this posterior sampler provides access to the entire
posterior distribution $\pi(F|x)$ of the statistic of interest $F$,
standard Bayesian methodology can be applied to quantify the
resulting uncertainty of any estimator $\hat{F}$, for example
by obtaining credible intervals in terms of the posterior $\pi(F|x)$.

\smallskip

An interesting special case of particular importance in practice
is the estimation of the \textit{entropy rate} 
$H$ of the underlying process,
$F(T,\theta) = H(T,\theta )$. The performance of all methods 
discussed above, with emphasis on the estimation 
of the entropy rate, is illustrated through simulated experiments 
and real-world applications in Section~\ref{s:entropy}. 

\section{Theoretical results} 
\label{theory}

Using the branching process representation of 
Proposition~\ref{p:branching_post},
we show how to derive precise
results on the asymptotic behaviour
of the BCT posterior $\pi(T|x)$ on model space,
and provide an explicit, useful expression
for the posterior predictive distribution.

\smallskip

Let $\{X_n\}$ be a variable-memory chain with model 
$T\in\mathcal {T}(D)$. 
The specific model $T$ that describes the chain is typically
not unique, for the same reason, e.g., that every i.i.d.\
process can also trivially be described as a first-order
Markov chain: Adding $m$ children to any leaf
of $T$ which is not at maximal depth, and giving each of them
the same parameters as their parent, leaves the distribution
of the chain unchanged.

The natural main goal in model selection is to identify
the ``minimal'' model, i.e., the smallest model that can
fully describe the distribution of the chain. A model $T\in\mathcal {T}(D)$ is called {\em minimal} 
if every $m$-tuple of leaves $\{sj\;;\;j=0,1,\ldots,m-1\}$
in $T$
contains at least two with non-identical parameters,
i.e., there are $j\neq j'$ such that $\theta_{sj}\neq\theta_{sj'}$.
It is easy to see that every $D$th order 
Markov chain $\{X_n\}$ has a unique minimal model 
$T^*\in\mathcal {T}(D)$. 

\smallskip

A variable-memory chain $\{X_n\}$
with model $T\in\mathcal {T}(D)$ 
and with associated parameters $\theta=\{\theta_s ; s\in T\}$
is {\em ergodic},
if the corresponding first-order chain 
\mbox{$\{Z_n:=X_{n-D+1}^n ; n\geq 1\}$} 
taking values in $A^D$
is irreducible 
and aperiodic.
In order to avoid uninteresting technicalities, 
in most of our results we will assume that
the data are generated by a {\em positive-ergodic}
chain $\{X_n\}$, namely 
that all its parameters $\theta_s (j)$ are nonzero,
so that its unique stationary
distribution $\pi$ gives strictly positive probability 
to all finite contexts $s$.

\subsection{Posterior consistency and concentration} \label{post_conc}

Our first theorem is a strong consistency result, which states
that, if the data $x=x_{-D+1} ^ n $ are generated by an ergodic chain with minimal model $T^*\in\mathcal {T}(D)$, then the model posterior asymptotically 
almost surely (a.s.) concentrates on $T^*$. 
Theorem~\ref{post_conc_thm} both strengthens
and generalises a weaker result on the asymptotic
behaviour of the MAP model established in 
\cite[Theorem~8]{willems-shtarkov-tjalkens:preS}.

\begin{thm} 
\label{post_conc_thm}
Let $X_{-D+1}^n=(X_{-D+1},\ldots,X_0,X_1,\ldots,X_n)$ be 
a time series generated
by a positive-ergodic, variable-memory chain 
$\{ X_n \}$ 
with minimal 
model $T^*\in\mathcal {T}(D)$. For any 
value of the prior hyperparameter $\beta\in(0,1)$, 
the posterior distribution 
over models concentrates on~$T^*$, i.e.,
\begin{equation*}
 \pi(T^*|X_{-D+1}^n) \to 1, \quad \text{a.s., as } n \to \infty .
\end{equation*}
\end{thm}

\begin{proof}
Recalling the posterior representation in~(\ref{post_alt_expr}),
it suffices to show that, as $n\to\infty$,
$P_{b,s} \to 0 $ a.s.\ for all internal nodes of $T^*$, and $P_{b,s} \to 1$ 
a.s.\ for all leaves of $T^*$. These two claims are established in 
Lemmas~\ref{internal} and~\ref{leaf} below.
\end{proof}

We first recall the following simple bounds on the estimated 
probabilities $P_{e,s}$; see
\cite{krichevsky1981performance,xie-barron:00}
and \cite[Ch.~1]{catoni:04}.

\begin{lem}
\label{lem:KT}
For every node $s$ with count vector $a_s=(a_s(0),a_s(1),\ldots,a_s(m-1))$
and $M_s=a_s(0)+\cdots+a_s(m-1)$, 
the estimated probabilities $P_{e,s}$ of~{\em (\ref{eq:Pe})}
satisfy:
\begin{align}
\log P_{e,s}\geq&
\sum_{j=0}^{m-1}a_s(j)\log\frac{a_s(j)}{M_s}
-\frac{m-1}{2}\log M_s-\log m;
\label{eq:PeLB}\\
\log P_{e,s}\leq&
\sum_{j=0}^{m-1}a_s(j)\log\frac{a_s(j)}{M_s}
-\frac{m-1}{2}\log\frac{M_s}{2\pi}-\log\frac{\pi^{m/2}}{\Gamma(m/2)}.
\label{eq:PeUB}
\end{align}
\end{lem}

[Throughout the paper, $\log\equiv\log_e$ denotes the natural logarithm.]
Lemma~\ref{internal} 
is a generalisation of 
\cite[Lemma~12]{jiao2013universal}.

\begin{lem}\label{internal}
Under the assumptions of Theorem~\ref{post_conc_thm}, 
for every internal node $s$ of $T^*$,
the branching probability
$P_{b,s} \to 0 $, a.s., as $n \to \infty$.
\end{lem}

\begin{proof}
As the complete proof is quite involved, 
only the main and more interesting part of the
argument is given here, 
with the remaining details given in Section B.1 
of the supplementary~material. 
We begin by observing that,
\begin{align}
P_{b,s} = \frac{\beta P_{e,s}}{P_{w,s}} = \frac{\beta P_{e,s}}{\beta P_{e,s}+ (1-\beta ) \prod _j P_{w,sj}} &= \frac{1}{1+{(1-\beta)}/{\beta} \ { \prod _j P_{w,sj}}/ {P_{e,s}}} \label {pbs_frac} \\ 
&\leq \frac{1}{1+ c_0 \ { \prod _j P_{e,sj}}/ {P_{e,s}}},
\end{align}
for some constant $c_0$,
where in the last step we used that either
$P_{w,sj} \ge \beta P_{e,sj}$ or $P_{w,sj} = P_{e,sj}$.
Therefore, it suffices to show that
$P_{e,s} \big / \prod _ j P_{e,sj} \to 0$, a.s.,
as $n\to\infty$.

\smallskip

Let $s$ be a fixed finite context, and let $X$ and $J$ denote the
random variables corresponding to the symbols that follow
and precede $s$, respectively, under the stationary distribution
of $\{X_n\}$.
Using Lemma~\ref{lem:KT} and the ergodic theorem for Markov chains,
it is shown in Section B.1 
of the supplementary material that,
\begin{equation} \label{mutual_info}
\log P_{e,s} - \sum _ j \log P_{e,sj} = - n I(X;J | s) \pi(s) + o(n), 
\quad \text{a.s.},
\end{equation}
where $I(X;J | s)$ is the conditional mutual information 
between $X$ and $J$ given $s$ \citep[Ch.~2]{cover1999elements}.
This mutual information is always nonnegative, and it is zero 
if and only if
$X$ and $J$ are conditionally independent given $s$.

For any internal node $s$ that is a parent of leaves of $T^*$, 
the minimality of $T^*$ implies that $X$ and $J$ are not 
conditionally independent given $s$, as there exist $j\neq j'$ such 
that~$\theta_{sj}\neq\theta_{sj'}$, so that $\theta_{sj}$ depends on $j$. 
Therefore, $I(X;J | s)>0$ and $\pi(s) >0 $ by assumption, 
so (\ref{mutual_info}) implies that $  \log P_{e,s} - \sum _ j \log P_{e,sj}  
\to - \infty$, a.s., as required.

\smallskip

For the general case of internal nodes that may not be parents of leaves,
a simple iterative argument is given in Section~B.1 of supplementary material
establishing the same result in that case as well, and 
completing the proof of the lemma.
\end{proof}

\begin{lem}\label {leaf}
Under the assumptions of Theorem~\ref{post_conc_thm}, 
for every leaf $s$ of $T^*$,
the branching probability $P_{b,s} \to 1 $, a.s.,
 as $n \to \infty$.
\end{lem}

\begin{proof}
As with the previous lemma, we only give an outline of the main
interesting steps in the proof here; complete details are provided
in Section~B.2 of the supplementary material.

\smallskip

By definition, for any leaf $s$ of $T^*$ (and also for any `external'
node $s$, that is, any context~$s$ not in $T^*$),
we have $I(X;J | s)=0$ because of conditional independence.
Therefore, we need to consider the higher-order terms in the 
asymptotic expansion of~(\ref{mutual_info}).
Using Lemma~\ref{lem:KT}, the ergodic theorem, and the law of 
the iterated logarithm~(LIL) for Markov chains, it is shown 
in Section~B.2 of the supplementary material that here,
\begin{equation} \label{47}
\sum _ j \log P_{e,sj} - \log P_{e,s} \leq - \frac {(m-1)^2} {2} \log n + O(\log \log n), \quad \text {a.s.}
\end{equation}
This implies that $\sum _ j \log P_{e,sj} - \log P_{e,s} \to - \infty$, 
so that  $\prod _j P_{e,sj}/ {P_{e,s}} \to 0$, a.s.

\smallskip

The last step of the proof, namely that for any leaf $s$ of $T^*$,
$\prod _j P_{e,sj}/ {P_{e,s}} \to 0$ also implies that  
$\prod _j P_{w,sj}/ {P_{e,s}} \to 0$, a.s., 
so that, by~(\ref{pbs_frac}), $P_{b,s}\to 1$, a.s., 
is given in Section~B.2 of the 
supplementary material.
\end{proof}

Our next result is a refinement of Theorem~\ref{post_conc_thm}, 
which characterises the rate at which the posterior
probability of $T^*$ converges to~1.

\begin{thm}
\label{t:postrate}
Let $X_{-D+1}^n$ be a time series generated by 
a positive-ergodic
variable-memory chain 
$\{ X_n \}$ 
with minimal model $T^*\in\mathcal {T}(D)$. For any 
value of the prior hyperparameter $\beta\in(0,1)$
and any $\epsilon>0$, we have, as $n\to\infty$:
\begin{equation*}
\pi (T ^ * | X_{-D+1}^n ) = 1 - O\left (  n^{-\frac{(m-1)^2}{2}+\epsilon} \right ), 
\quad \text {a.s.}
\end{equation*}
\end{thm}

The proof of Theorem~\ref{t:postrate} is given in Section~B.3 of the 
supplementary material. In fact, as discussed at the end of Section~B.3,
the proof also reveals that a stronger statement
can be made about the rate in the case of full $D$th order Markov chains:

\begin{cor}
If $\{X_n\}$ is a genuinely $D$th order chain in that its minimal
model $T^*$ is the complete tree of depth $D$, then its posterior
probability $\pi(T^*|X_{-D+1}^n)$ almost surely
converges to~1 at an exponential rate.
\end{cor}

\subsection{Out-of-class modelling}

In this section we consider the behaviour
of the posterior distribution $\pi(T|x)$ 
on models~$T\in{\mathcal T}(D)$ when the time 
series $x$ is {\em not} generated by a chain 
from the model class ${\mathcal T}(D)$, but from
a general stationary and ergodic process $\{X_n\}$
with possibly infinite memory.
We first
give an explicit description of 
the ``limiting'' model $T_\infty\in{\mathcal T}(D)$ 
on which the posterior $\pi(T|x)$ concentrates
when the observations are generated by a general 
process outside ${\mathcal T}(D)$,
and then we give conditions under which $T_\infty$
is structurally ``as close as possible'' to the
true underlying~model.

\medskip

\noindent
{\bf Description of $T_\infty $. }  
Recall that, for any context $s$, we write
$X$ and $J$ for the random variables corresponding to the 
symbols that follow and precede $s$, respectively.
The limiting tree $T_\infty\in{\mathcal T}(D)$ corresponding 
to a general stationary process $\{X_n\}$ with 
values in $A$ can be constructed via the following
procedure:

\begin{itemize}
\item
Take $T _ \infty$ to be the empty tree.
\item
Starting with the nodes at depth $d=D-1$, 
for each such $s$, if $I(X;J|s)>0$, then add $s$ to $T _ \infty$ 
along with all its children and all its ancestors;
that is, add the complete path from the root $\lambda$ to the children
of $s$. 
\item
After all nodes $s$ at depth $d=D-1$ have been examined, 
examine all possible nodes $s$ at depth $d=D-2$ that are 
not already included in $T_\infty $, and repeat the same process. 
\item
Continue recursively towards the root, until all nodes 
at all depths $0\leq d \leq D-1$ have been examined. 
\item
For any node already in $T_\infty$ 
at depth $d \leq D-1$,
such that only some but not all $m$ of its 
children are included in~$T_\infty$, 
add the missing children to $T_\infty$ so that it becomes
proper.
\item
Output $T_\infty$.
\end{itemize}

In order to state our results we need to impose two additional
conditions on the underlying data-generating process.
Suppose $\{X_n\}$ is stationary. Without loss of generality
(by Kolmogorov's extension theorem) we may consider the two-sided
version of the process, $\{X_n\;;\;n\in{\mathbb Z}\}$.
Its $\alpha$-mixing coefficients 
\citep{ibragimov:62,philipp1975almost} are defined as,
\begin{equation}
\alpha _ n = 
\sup_ {A \in \mathcal {A},\; B \in \mathcal {B}} 
|  \mathbb {P} (A \cap B ) -\mathbb {P} (A) \mathbb {P} (B)   |,
\end{equation}
where $\mathcal {A} $ and $\mathcal {B}$ denote the 
$\sigma$-algebras,
$\sigma (\ldots,X_{-1},X_0)$ 
and $\sigma (X _ {n} , X_ {n+1} , \ldots)$, respectively.
We will need a mixing condition and a positivity condition for our
results:
\begin{equation}
\sum _ {n=1} ^ \infty \alpha _ n ^ {\delta / (2 + \delta ) } < \infty
\;\mbox{for some}\;\delta>0,
\quad\mbox{and}\quad {\mathbb P}(X_0^{D}=x_0^D)>0\;\mbox{for all}
\;x_0^D\in A^{D+1}.
\label{eq:assume}
\end{equation}

\begin{thm}
\label{t:out}
Let $X_{-D+1}^n$ be a time series generated by a stationary ergodic process~$\{X_n\}$ satisfying the assumptions~{\em (\ref{eq:assume})},
and let $T_\infty\in {\mathcal T}(D)$ be given
by the above construction. Then,
for any value of the prior hyperparameter $\beta\in(0,1)$ we have:
\begin{equation*}
 \pi(T_ \infty |X_{-D+1}^n) \to 1, \quad \text{a.s., as } n \to \infty .
\end{equation*}
\end{thm}

The proof of Theorem~\ref{t:out} follows along exactly the same lines
as the earlier proof of Theorem~\ref{post_conc_thm}. Instead of the
ergodic theorem for Markov chains \citep[p.~92]{chung1967markov} we
now use Birkhoff's ergodic theorem \citep[Ch.~6]{breiman:book},
and instead of the LIL for Markov chains we apply the general
LIL for functions of blocks of an ergodic process, which follows,
as usual, from the almost-sure invariance principle
\citep{philipp1975almost, rio1995functional,zhao2008law}.
The mixing condition in~(\ref{eq:assume}) was chosen
as one of the simplest ones that guarantee this general
version of the LIL.

Regarding the overall structure of the proof,
all the earlier asymptotic expansions 
of the branching probabilities still remain valid,
including~(\ref{mutual_info}) and~(\ref{47}).
The only possible difference might be at the boundary conditions 
that are required as a starting point for the iterative argument 
in the proof of Lemma~\ref{internal}, since the actual ``leaves''
of the true underlying model of $\{X_n\}$
are not necessarily at depth~$d \leq D$ here. But, as before, 
all branching probabilities $P_{b,s}$ tend 
either to 0 or to 1, depending on whether the mutual 
information condition that appears in the description
of $T_\infty$ holds or not.
Specifically, starting from nodes~$s$ at depth $d=D-1$, 
if $I(X;J|s)=0$ then we are in the same situation 
as in Lemma~\ref{leaf}, so that  $P_{b,s} \to 1$, 
and all children of $s$ are pruned. On the other hand,
if $I(X;J|s)>0,$ then $P_{b,s} \to 0$, and by the same 
iterative argument, 
$P_{b,u} \to 0$ for all ancestors $u$ of node $s$ as well.

Analogous comments apply to the proof of Theorem~\ref{t:outrate}, which is again a refinement characterising the rate at which the posterior probability of $T_ \infty $ converges to 1.

\begin{thm}
\label{t:outrate}
Let $X_{-D+1}^n$ be a time series generated by a stationary ergodic process~$\{X_n\}$ satisfying the assumptions~{\em (\ref{eq:assume})}, and
let $T_\infty$ be its limiting model in ${\mathcal T}(D)$. Then,
for any $\beta\in(0,1)$ and 
any $\epsilon>0$, as $n\to\infty$ we have:
\begin{equation*}
\pi (T_\infty | X_{-D+1}^n ) = 1 - O\left (  n^{-\frac{(m-1)^2}{2}+\epsilon} \right ), \quad\mbox{a.s.}
\end{equation*}
\end{thm}

In general, it is natural to expect that $T_\infty$ should be 
``as close as possible'' in some sense 
to the true underlying model $T^*$,
and this is indeed what is most often observed in applications:
$T_\infty$ being the same as $T^*$ truncated to depth $D$.

\newpage

But this is not always the case. For example,
recall the 3rd order chain $\{X_n\}$ considered in Example~5.2
of~\cite{our},
also described as having a ``bimodal posterior'' in Section~\ref{s:entropy}.
There, $T^*$ is the complete $m$-ary tree of depth~3
(with $m=6$),
but $X_n$ depends on $(X_{n-1},X_{n-2},X_{n-3})$
only via $X_{n-3}$. For that reason, the limiting model 
$T_\infty$ in ${\cal T}(D)$ with $D=1$ or $D=2$ 
is not 
$T^*$ truncated at depth $D$, but rather the empty tree
$\{\lambda\}$ consisting of only the root node $\lambda$.

\smallskip

Fortunately, we can read
a simple necessary and sufficient condition 
for the ``expected'' behaviour to occur from
the definition of $T_\infty$ itself.
Let $\{X_n\}$ be a stationary and ergodic process on 
the finite alphabet $A$. From the 
description of \cite{csiszar2006context}, 
it is easy to see that there is a unique
minimal context tree model $T^*$ for
$\{X_n\}$ of possibly infinite depth. 
Let $T ^*  _  { | D} $ be $T^*$ truncated at 
depth $D$, and write $N_{D-1} (T^*)$ for the set of all 
internal nodes of $T^*$ at depth $d=D-1$ whose children exist 
in $T^*$ (at depth $D$) but are not leaves of $T^*$.

\begin{cor}\label{cor}
The limiting model 
$T_ \infty =  T ^*  _  { | D}$ if and only if 
$I(X;J|s) > 0 $ for all the nodes $s \in N_{D-1} (T^*)$.
\end{cor}

\begin{proof}
The result follows directly from the definition of 
$T_\infty$ combined with the observation
that the condition $I(X;J|s) > 0 $ 
is already satisfied for all nodes $s$ of $T^*$ at depth $d=D-1$ whose 
children are leaves of $T^*$.
\end{proof}

\subsection{The posterior predictive distribution}

The branching process representation of the posterior can
also be used to facilitate practically useful computations. 
In Proposition~\ref{pred_distr}, an exact expression is given for 
the posterior predictive distribution 
$P\left (x_{n+1}|x_{-D+1}^ n \right )$ in terms of
the branching probabilities~$P_{b,s}$.

\begin{prop}\label{pred_distr}
The posterior predictive distribution is given by,
\begin{equation}\label{48}
P \left (x_{n+1} | x _{-D+1 } ^ n \right  ) = \sum _ {i=0} ^ D \left ( \frac{a_{s ^ {(i)}} \left ( x_ {n+1}\right ) + 1/2} {M_{s ^ {(i)}} + m/2}\right ) \gamma_i  ,
\end{equation}
where, for $0 \leq i \leq D$, the string $s  ^ {(i)}$ is 
the context of length $i$ preceding $x_{n+1}$, and $\gamma_i$ is 
the posterior probability that node $s  ^ {(i)}$ is a leaf, given by,
\begin{equation}\label{post_s_leaf}
\gamma _ i\!=\!
\left\{
\begin{array}{ll}
P_{b, \lambda},  \; \; \;  & i=0 ,   \\
\prod _{k=0} ^ {i-1} \left ( 1- P _ {b , s  ^ {(k)} } \right )   P_{b,s  ^ {(i)} } ,  \; \; \; & 1 \leq i \leq D-1 ,   \\
\prod _{k=0} ^ {D-1} \left ( 1- P _ {b , s  ^ {(k)} } \right ) ,    \; \; \; &  i=D . 
\end{array}\!\!
\right.\!\!
\end{equation}
\end{prop}

\begin{proof}
Writing $x=x_ {-D+1} ^ n $, 
the posterior predictive distribution
can be expressed as,
\begin{equation} \label{pred_distr_1}
P \left (x_{n+1} | x \right  )  = \sum _ {T \in \mathcal {T} (D)}  P \left (x_{n+1} | T , x  \right  )  \pi (T | x ) . 
\end{equation}
For any tree ${T \in \mathcal {T} (D)} $, 
exactly one of the contexts $s  ^ {(i)}$, $0 \leq i \leq D$,  is a 
leaf of the tree.
For every $0 \leq i \leq D$, define the subset $\mathcal {T} _ i (D) 
\subset \mathcal T(D) $ to be the collection of trees~$T \in \mathcal {T} (D)$ 
such that the context of $x_{n+1}$ that is a leaf of $T$ is $s  ^ {(i)}$; 
these $\mathcal {T} _ i (D) $ are disjoint and their union is~$\mathcal {T} (D)$. 

\newpage 

\noindent The key observation here is that $P (x_ {n+1} | T , x)  $ is the same for 
all trees $T \in \mathcal {T} _ i (D) $, since,
\begin{align}
P (x_ {n+1} | T , x)  =&  \int _ \theta P (x_ {n+1}  | T , \theta ,  x)  \pi (\theta | T , x )  \ d \theta \nonumber \\
= & \int _ { \theta_{s^{(i)}} } P (x_ {n+1}  | T , \theta _ {s ^ {(i)}} ,  x)  \pi (\theta _ {s ^ {(i)}}  | T , x )  \ d  \theta _ {s ^ {(i)}}\nonumber \\
 = &  \int _ { \theta _ {s ^ {(i)}} } \theta _ {s ^ {(i)}} \left ( x_ {n+1}\right )  \pi (\theta _ {s ^ {(i)}}  | T , x )  \ d  \theta _ {s ^ {(i)}} =  \frac{a_{s ^ {(i)}} \left ( x_ {n+1}\right ) + 1/2} {M_{s ^ {(i)}} + m/2} ,
\end{align}
where we used the full 
conditional density of the parameters
in~(\ref{post_theta}).
So, from~(\ref{pred_distr_1}),
\begin{align*}
P \left (x_{n+1} | x \right  )  =  \sum _ {i=0} ^ D   \sum _ {T \in \mathcal {T} _ i (D)} \hspace*{-0.1 cm} P \left (x_{n+1} | T , x  \right  )  \pi (T | x ) 
 =  \sum _ {i=0} ^ D   \frac{a_{s ^ {(i)}} \left ( x_ {n+1}\right ) + 1/2} {M_{s ^ {(i)}} + m/2}  \sum _ {T \in \mathcal {T} _ i (D)} \hspace*{-0.1 cm}  \pi (T | x ) ,
\end{align*}
which completes the proof upon noticing that
the last sum $\sum _ {T \in \mathcal {T} _ i (D)}   \pi (T | x ) $ is exactly 
the posterior probability that node $s  ^ {(i)}$ is a leaf, 
namely, $\gamma _ i $ as in~(\ref{post_s_leaf}).
\end{proof}

\vspace*{-0.4 cm}

\section{Experimental results} \label{exp}

\vspace*{-0.1 cm}

Being able to obtain 
exact i.i.d.\ samples from the posterior 
is generally more desirable and typically
leads to more efficient estimation than using
approximate MCMC samples.
In Section~\ref{s:MCMC} we offer empirical evidence
justifying this statement in the present setting
through 
a simple simulation example. 
Then in Section~\ref{s:entropy}
we present the results of a careful empirical
study of the natural entropy estimator induced
by the BCT framework, compared against
a number of the most common 
alternative estimators,
on three simulated and three real-world data sets.

\vspace*{-0.1 cm}

\subsection{Comparison with MCMC}
\label{s:MCMC}

%\vspace*{-0.1 cm}

Consider 
$n=1000$ observations generated from 
a 5th order, ternary chain, with model
given by the 
context tree of Figure~\ref{tree} in Section~\ref{bct}
(the values of the parameters $\theta=\{\theta_s;s\in T\}$
are given in Section~C of the supplementary 
material).
A simple and effective convergence diagnostic 
here (which can also be viewed as an example
of an estimation problem)
is the examination of the frequency with which the MAP model,
$T_1^*$,
appears in the i.i.d.\ or the MCMC sample trajectory.
The model $T_1^*$ can be identified by the BCT algorithm 
and its posterior probability $\pi (T _ 1 ^ *|x)$
can be computed, as in \citep{our}.

\begin{figure}[!ht]
\centering
\vspace*{-0.4 cm}
 \includegraphics[width= 0.6 \linewidth]{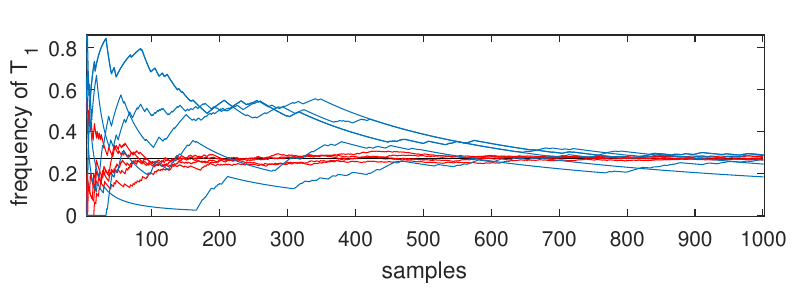}
\vspace*{-0.3 cm}
\caption{Frequency of $T _ 1 ^ *$. Blue: MCMC estimates. Red: 
i.i.d.\ estimates. In each case, the five graphs correspond to 
five independent repetitions of the experiment with 
\mbox{$N=1000$} simulated samples. The horizontal 
line is the limiting frequency, $\pi (T _ 1 ^ *|x)$.}
\label{ternt1}
\end{figure}

As shown in Figure~\ref{ternt1}, the estimates
based on the random-walk MCMC sampler of \cite{our}
and on the i.i.d.\ sampler of Section~\ref{s:sampling}
both
appear to converge quite quickly, 
with the corresponding MCMC estimates converging significantly
more slowly. 
In
50 independent repetitions of the same 
experiment (with $N=1000$ simulated
samples in each run), the estimated variance of the
MCMC estimates (0.0084) was found to be
larger than that 
for the i.i.d.\ estimates
($1.4 \times 10 ^ {-4}$),
by a factor of around~60.

Figure~\ref{tern_trace} shows the trace plots 
\citep{roy2020convergence}
obtained from
$N=10000$ simulated samples from the MCMC and i.i.d.\ samplers, which can be used to monitor the log-posterior in each case.
It is immediately evident that the i.i.d.\ sampler is more efficient
in exploring the effective support of the posterior.

\begin{figure}[!ht]
\vspace*{-0.3 cm}
\centering
 \includegraphics[width= 3.7in]{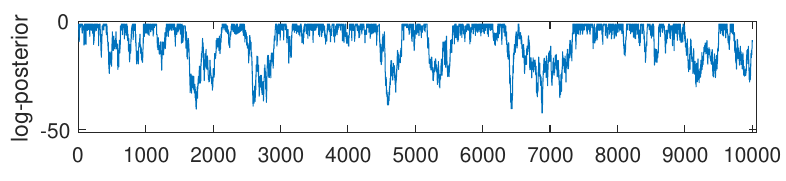}

\vspace*{-0.2 cm}

\hspace{0.12in}
\includegraphics[width= 3.63in, height=0.92in]{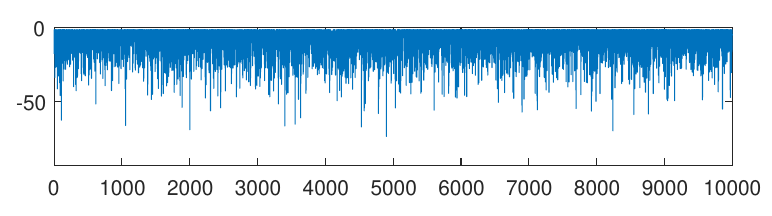}

\vspace*{-0.2 cm}

\caption{Trace plots showing the log-posterior,
$\log \pi (T  ^{(i)}|x)$,
at each iteration.
Top:~MCMC samples.
Bottom: i.i.d.\ samples.}
\label{tern_trace}
\vspace*{-0.2 cm}
\end{figure}

As expected,
the i.i.d.\ sampler has
superior performance compared to the MCMC sampler, 
both in terms of estimation and in terms of mixing. 
Also, although the two types of samplers
have comparable complexity in terms of computation time 
and memory requirements, the structure of the i.i.d.\
sampler is much simpler, giving a much easier implementation.
In view of these observations, in the following section we only 
employ the i.i.d.\ sampler for the purposes of
entropy estimation.

\vspace*{-0.2 cm}

\subsection{Entropy estimation}
%%%%%%%%%%%%%%%%%%%%%%%%%%%%%%%%%%%%%%%%%%%%%%%%%%%%%%%%%%%%%%%
\label{s:entropy}

\vspace*{-0.1 cm}

Estimating the entropy rate from empirical data
-- in this case, a discrete time series --
is an important and timely problem
that has received a lot of attention in the recent
literature, in connection with questions
in many areas including
neuroscience~\citep{timme:18},
natural language modelling~\citep{willems:16},
animal communication~\citep{kershenbaum:14},
and cryptography~\citep{simion:20}, 
among others;
see, e.g., the recent
literature reviews by~\cite{verdu2019empirical}
and~\cite{feutrill:21}. The well-known
difficulties of entropy estimation
stemming from the  nonlinear nature
of the entropy rate functional and its 
dependence on the entire process distribution
are discussed in the references listed
above.

For a general process $\{X_n \}$ on a finite alphabet, 
the {\em entropy rate} $\bar{H}$ 
is defined as the limit $\bar{H} = \lim _ {n \to \infty} (1/n)H (X_1^n)$,
whenever the limit exists, where $H(X_1^n)$ denotes the
usual Shannon entropy (in nats rather than bits, as we take logarithms
to the base $e$) of the discrete random vector
$X_1^n$. For an ergodic, first-order Markov 
chain $\{X_n\}$,
$\bar{H}$ can be expressed as,
\begin{equation}\label{rate_mc}
\bar{H} = - \sum _{i,j \in S}  \pi (i) P_{ij} \log P _{ij}, 
\end{equation}
where $S$ is the state space of $\{X_n\}$,
and $(P_{ij})$ and $(\pi(i))$ denote its
transition matrix and its
stationary distribution, respectively.

\newpage

An analogous formula can be written 
for the entropy rate of any ergodic variable-memory 
chain with model $T\in{\mathcal T}(D)$, by viewing it as a full $D$th order chain and considering blocks of length $(D+1)$, as usual;
cf.\ \cite{cover1999elements}. This means that
$\bar{H}$ can be expressed as an explicit 
function $\bar{H}=H(T ,\theta)$ of the model
and~parameters.

Therefore, given a time series $x$,
using the MC sampler of Section \ref{section3.3}
to produce i.i.d.\ samples $(T^{(i)},\theta^{(i)})$
from $\pi(T,\theta|x)$, we can obtain i.i.d.
samples $H^{(i)}=H(T^{(i)},\theta^{(i)})$ from the
posterior $\pi(\bar{H}|x)$ of the entropy rate.
The calculation of each 
$H^{(i)}=H(T^{(i)},\theta^{(i)})$ is straightforward
and only requires the computation of the stationary
distribution $\pi$ of the induced first-order chain that
corresponds to taking blocks of size [depth$(T^{(i)})+1$].
The only potential difficulty is if either the
depth of $T^{(i)}$ or the alphabet size $m$ 
are so large that the computation of $\pi$ becomes
computationally expensive. In such cases, $H^{(i)}$
can be computed approximately by including an 
additional Monte Carlo step: Generate a sufficiently
long random sample $Y_{-D+1}^M$ from the chain 
$(T^{(i)},\theta^{(i)})$, 
and calculate:
\begin{equation}
H^{(i)}\approx-\frac{1}{M}\log P(Y_1^M|Y_{-D+1}^0, T^{(i)},\theta^{(i)}).
\label{eq:MCMCMC}
\end{equation}
The ergodic theorem and the central limit theorem
for Markov chains 
\citep{chung1967markov,meyn2012markov}
then guarantee the accuracy
of~(\ref{eq:MCMCMC}).

In the remainder of this section, the BCT estimator 
(with maximum model depth $D=10$) is compared with 
the state-of-the-art approaches,
as identified by \cite{gao2008estimating} 
and \cite{verdu2019empirical} and summarised below.
The BCT estimator is found to generally give the most reliable estimates 
on a variety of different types of simulated and real-world
data. Moreover, compared 
to most existing approaches that give simple point estimates 
(sometimes accompanied by confidence intervals), 
the BCT estimator has the additional advantage 
that it provides the entire posterior distribution
$\pi(\bar{H}|x)$.

\smallskip

\noindent
{\bf Plug-in estimator. } 
Motivated by the definition of the entropy rate,
the simplest and one of the most commonly used
estimators of the entropy rate is the per-sample
entropy of the empirical distribution of $k$-blocks.
Letting $\widehat p_k (y_1 ^k ) $, $y_1^k\in A^k$,
denote the empirical distribution of $k$-blocks
induced by the data on $A^k$, the {\em plug-in}
or {\em maximum-likelihood} estimator is simply,
$\widehat H_ k = (1/k) H (\widehat p_k)$.
The main advantage of this estimator
is its simplicity. Well-known drawbacks
include its high variance due to undersampling,
and the difficulty
in choosing appropriate block-lengths $k$
effectively.

\smallskip

\noindent
{\bf Lempel-Ziv estimator. } 
Among the numerous match-length-based entropy estimators
that have been derived from the Lempel-Ziv family of 
data compression algorithms, we consider the
increasing-window estimator of \cite{gao2008estimating},
identified there as the most effective one.
For every position $i$ in the observed data,
let $\ell_i$ denote the length of the longest segment 
$x_i ^{i+\ell_i-1}$ starting at $i$ which also appears 
somewhere in the window~$x_{0} ^{i-1}$ preceding $i$. 
Writing $L_i = 1 + \ell_i$ for each $i$, 
the relevant estimator is, \vspace*{-0.1 cm}
$$
\widehat H_ {\text{LZ}} = \frac{1}{n} \sum _ {i=2} ^ n \frac{\log i }{L_i}. \vspace*{-0.1 cm}
$$
\noindent
{\bf CTW estimator. } This uses the prior predictive
likelihood $P(x)$ computed by the CTW algorithm,
to define $\widehat H _ { \text {CTW}} = -(1/n)\log P(x_1^n)$. 
This estimator was found by \cite{gao2008estimating} 
and \cite{verdu2019empirical} to achieve the best performance in practice. 
Its consistency and asymptotic normality follow easily from standard
results, and its (always positive) bias is of $O((\log n)/n)$, which
can be shown to be in a minimax sense as small as possible.
In all experiments we take the maximum depth of CTW to be $D=10$.

\newpage

\noindent
{\bf PPM estimator. } Using a different
adaptive probability assignment, $Q(x)$, this method
forms an estimate of the same type as the CTW estimator,
$\widehat H _ { \text {PPM}} = -(1/n)\log Q(x_1^n)$,
where prediction by partial matching (PPM)
\citep{cleary1984data} is used to fit the model 
that leads to $Q(x_1^n)$. 
We use the interpolated smoothing variant of PPM 
introduced by \cite{bunton1996line}, 
which is implemented in the R package 
available at: \url{https://rdrr.io/github/pmcharrison/ppm/}.

\medskip

\noindent
{\bf A ternary chain. } We consider the same
$n=1000$ observations generated from the 5th order,
ternary chain examined in Section~\ref{s:MCMC}. 
The entropy rate of this chain is~$\bar{H}=1.02$. 
In Figure \ref{post_evolution} we show MC estimates of the
prior distribution $\pi(\bar{H})$,
and of the posterior~$\pi(\bar{H}|x)$
based on $n=100$ and on $n=1000$ observations from the chain.
After $n=1000$ observations, the posterior is
close to a Gaussian  
with mean $\mu=1.005$ and standard deviation 
$\sigma = 0.017$. 
For each histogram $N=10^5$ i.i.d.\ samples were 
used, and in each
case (and in all subsequent examples), the vertical 
axis of the histograms shows the frequency of the bins
in the Monte Carlo sample.

\begin{figure}[!ht]
\vspace*{-0.3 cm}
\begin{subfigure}{0.32 \linewidth}
\hspace*{-0.1in}
\includegraphics[width= 1.09 \linewidth]{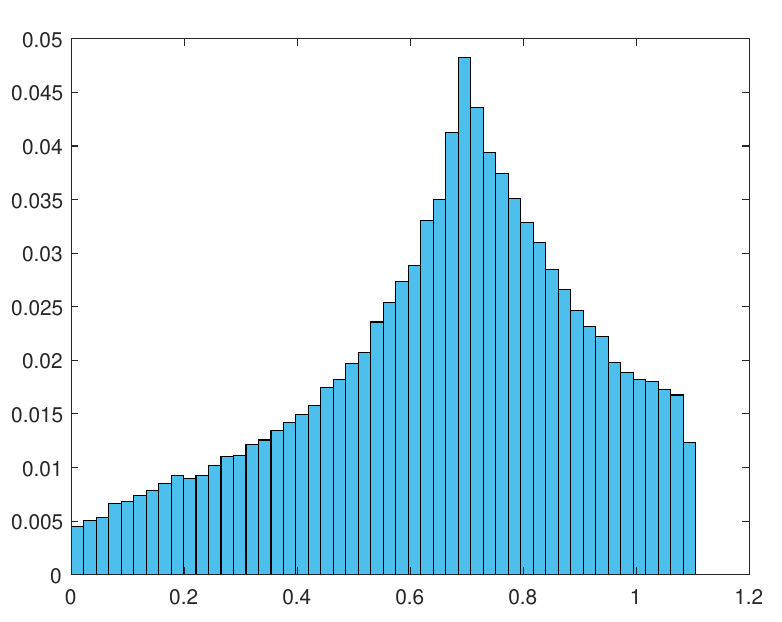}
\vspace*{-0.49 cm}
\caption{prior}
\end{subfigure}
\begin{subfigure}{0.32 \linewidth}
\hspace*{-0.03in}
 \includegraphics[width= 1.10 \linewidth]{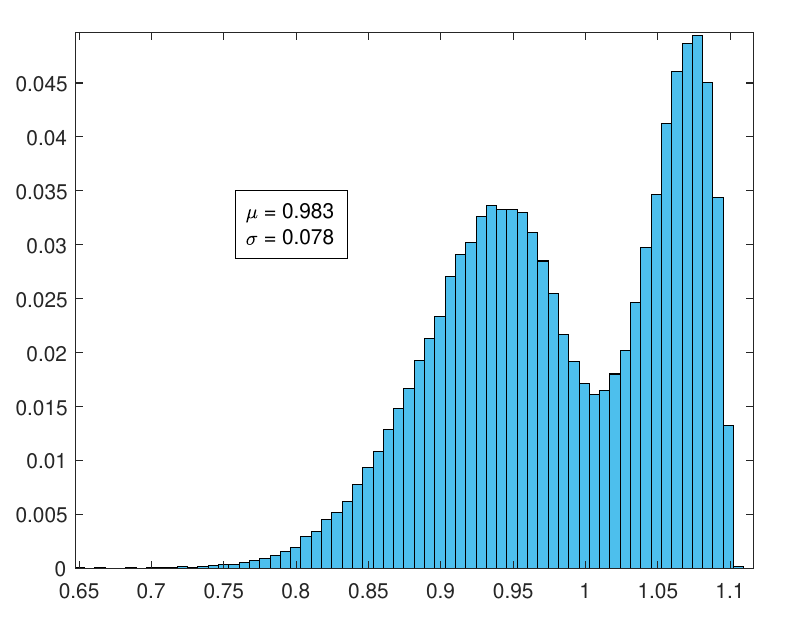}
\vspace*{-0.52 cm}
\caption{$n=100$}
\end{subfigure}
\vspace*{-0.52 cm}
\begin{subfigure}{ 0.32 \linewidth}
\hspace*{0.06in}
 \includegraphics[width= 1.05 \linewidth]{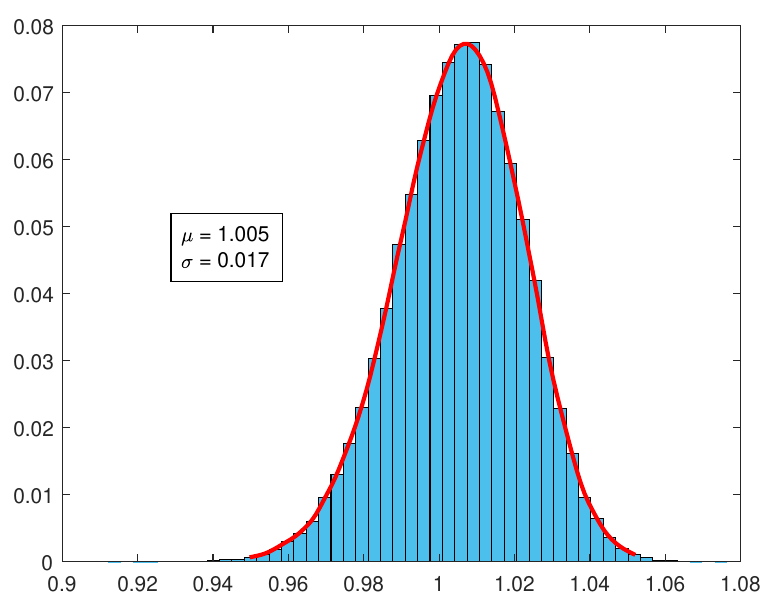}
\vspace*{-0.41 cm}
\caption{$n=1000$}
\end{subfigure}
\vspace*{0.3 cm}
\caption{Prior $\pi(\bar{H})$ and posterior 
$\pi(\bar{H}|x)$ of the 
entropy rate $\bar{H}$ with $n=100$ and $n=1000$ observations 
$x$.}
\label{post_evolution}
\end{figure}

\vspace*{-0.3 cm}

Figure \ref{tern_h_plots} shows the performance
of the BCT estimator compared with the other
four estimators described above, as a function
of the length $n$ of the available observations~$x$. 
For BCT we plot the posterior mean. For the plug-in
we plot estimates with block-lengths $k=5,6,7$. It is easily observed that the BCT estimator outperforms all the alternatives, and converges faster and closer to the true value of $ \bar H  $.

\begin{figure}[!ht]
\centering
\vspace*{-0.2 cm}
\includegraphics[width= 0.77 \linewidth]{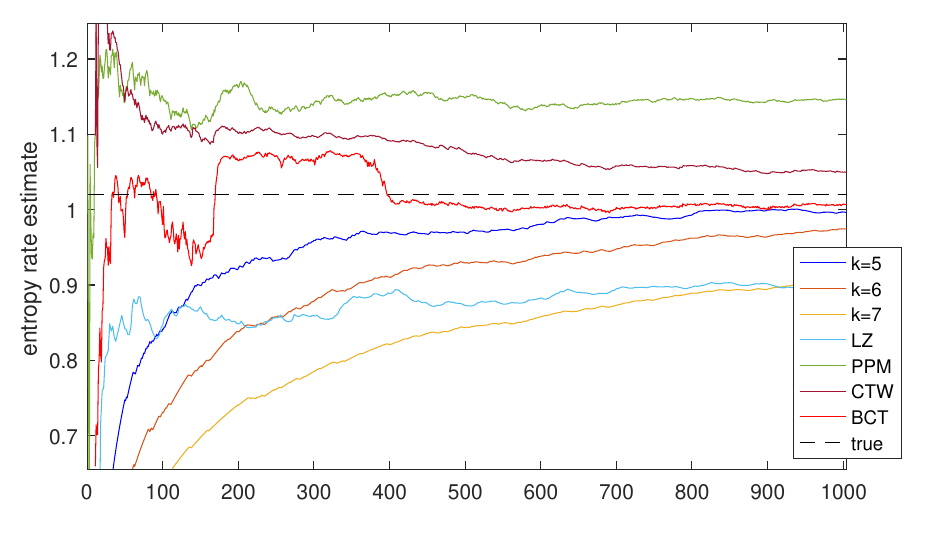}
\vspace*{-0.35 cm}
\caption{Entropy rate estimates for the 5th order ternary chain,
as the number of observations increases.}
\label{tern_h_plots} 
\vspace*{-0.2 cm}
\end{figure}

\noindent
{\bf A third order binary chain. } Here, we consider $n=1000$ observations 
generated from an example of a third order binary chain 
from \cite{berchtold2002mixture}. The underlying model is the 
complete binary tree of depth $3$ pruned at node $s=11$;
the tree model $T$ and the parameter values
$\theta=\{\theta_s;s\in T\}$ are given in Section~C
of the supplementary material. 
The entropy rate of this chain is $\bar{H}=0.4815$. 
Figure~\ref{y3_pots} shows the performance 
of all five estimators,
where 
the BCT estimator (which uses the posterior mean again) is found to have the best performance.
The histogram of the BCT posterior after $n=1000$ observations,
shown in Section~C of the supplementary material, is
close to a Gaussian with a
mean $\mu = 0.4806$ and a standard deviation~$\sigma = 0.0405$.

\begin{figure}[!ht]
\centering
\vspace*{-0.1 cm}
\includegraphics[width= 0.82 \linewidth]{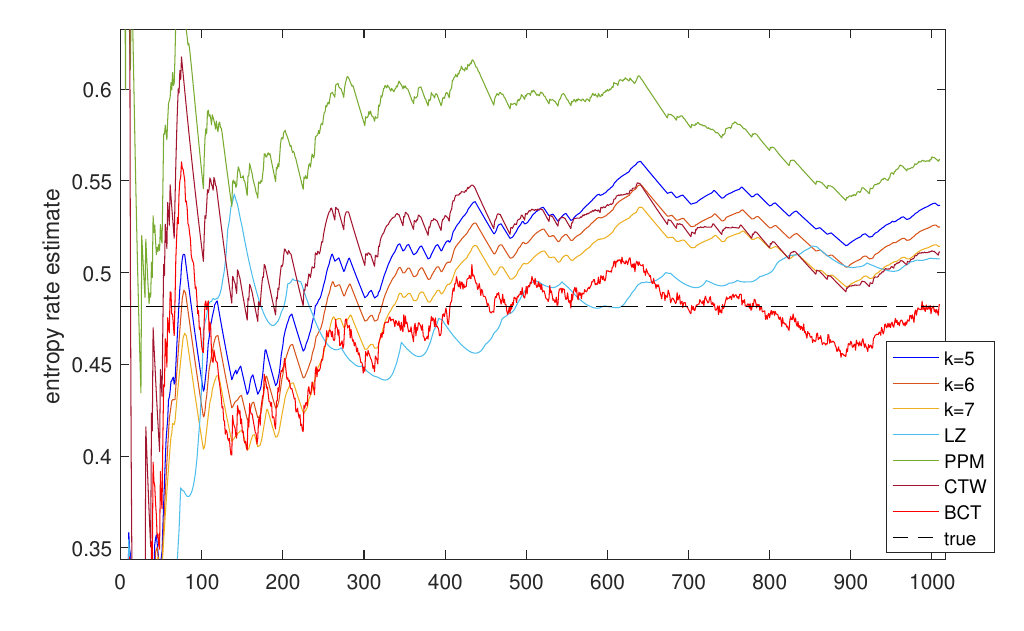}
\vspace*{-0.4 cm}
\caption{Entropy rate estimates for the third order binary chain,
as the number of observations increases.}
\label{y3_pots} 
\vspace*{-0.1 cm}
\end{figure}

% \vspace*{-0.4in}

\noindent
{\bf A bimodal posterior. } 
We re-examine a simulated time series $x$ 
from \cite{our}, which
consists of
$n=1450$ observations generated from 
a $3$rd order chain~$\{X_n\}$ with 
alphabet size $m=6$ and with the property that each $X_n$ depends 
on past observations only via $X_{n-3}$.
The complete
specification of the chain is given in Section~C 
of the supplementary material.
Its entropy rate is $\bar{H}=1.355$. 
An interesting aspect of this data set is that 
the model posterior is bimodal, with one mode
corresponding to the empty tree (describing i.i.d.\ observations)
and the other consisting of tree models of depth~3.

\smallskip

As shown in Figure \ref{bimodal_h}, the posterior
of the entropy rate is also bimodal here, 
with two separated approximately-Gaussian modes 
corresponding to each of the modes of the model posterior. 
The dominant mode is the one corresponding to models of depth 3; 
it has mean $\mu_1=1.406$, standard deviation $\sigma_1 =0.031$, 
and relative weight $w_1 = 0.91$. The second mode 
corresponding to the empty tree has mean $\mu_2 = 1.632$, standard 
deviation $\sigma_2 = 0.020$, and a much smaller weight 
$w_2 =1 -w_1 =  0.09$. In this case,
the mode of $\pi(\bar{H}|x)$ gives a more reasonable 
choice for a point estimate than the posterior mean.
Like in the previous two examples,
the BCT entropy estimator 
performs better than most benchmarks, 
as illustrated in
Section~C of the supplementary material. 

\newpage

\begin{figure}[!ht]
\begin{subfigure}{0.49 \linewidth}
 \includegraphics[width= 1 \linewidth]{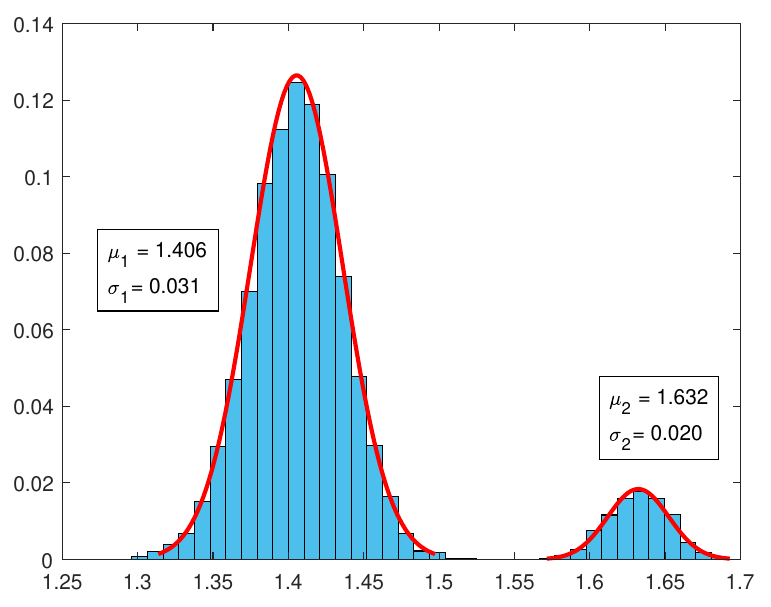}
\vspace*{-0.5 cm}
\caption{bimodal example}
\vspace*{0.02 cm}
\label{bimodal_h}
\end{subfigure}
\begin{subfigure}{0.49 \linewidth}
 \includegraphics[width= 1 \linewidth]{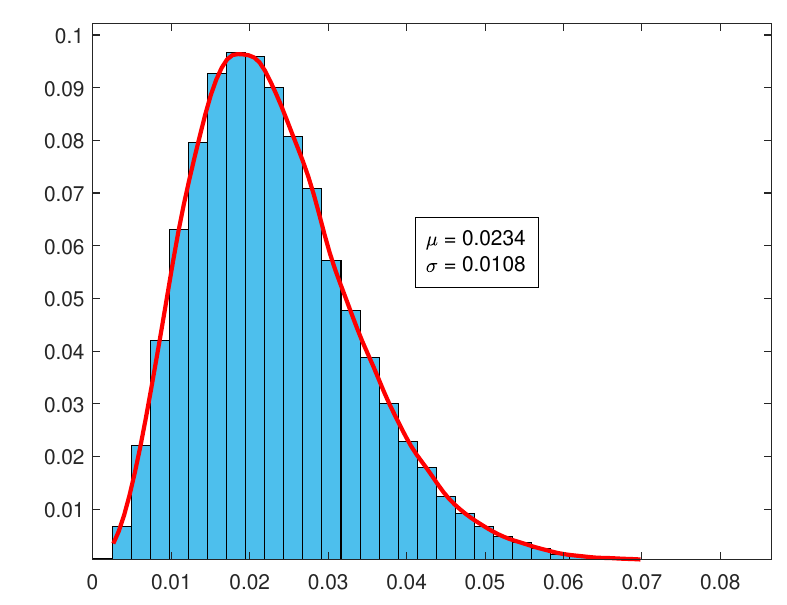}
\vspace*{-0.45 cm}
\caption{spike train}
\label{spike_h}
\end{subfigure}
\vspace*{-0.2 cm}
\caption{Histograms of the
posterior distribution $\pi(\bar{H}|x)$ of the entropy rate,
constructed from $N=10^5$ i.i.d.\ samples in each case.}
\vspace*{-0.4 cm}
\end{figure}

\medskip

\noindent
{\bf Neural spike trains. } 
We consider $n=1000$ binary observations from a spike train recorded from a single neuron in region V4 of a monkey's brain. The BCT posterior is 
shown in Figure \ref{spike_h}: Its mean is $\mu = 0.0234$, its standard deviation 
is $\sigma = 0.0108$, and is skewed to the right. This dataset is 
the first part of a long spike train of length $n=3,919,361$ 
from \cite{gregoriou2009high,gregoriou2012cell}.
Although there is no ``true'' value of the entropy rate here,
for the purposes of comparison we use the estimate obtained
by the CTW estimator (identified as the most effective method
by \cite{gao2008estimating} and \cite{verdu2019empirical})
when all $n=3,919,361$ samples are used, giving
$\bar{H}=0.0241$. 
The resulting estimates for all 
five methods (with the posterior 
mean given for BCT) are summarised in Table~\ref{spike_table}, verifying again that BCT outperforms all the other methods.

\begin{table}[!ht]
\centering
{\small
\begin{tabular}{cccccccccc}
\midrule 
 & ``True" & BCT & CTW & PPM &LZ & $k=2$ & $k=5$ & $k=10$ & $k=15$ \\
\midrule
$\widehat H $ & 0.0241 & \bf{0.0234} & 0.0249 & 0.0360 & 0.0559 & 0.0204 & 0.0204 & 0.0198 & 0.0187 \\
\bottomrule
\end{tabular}
}
\caption{Entropy rate estimates for the neural spike train.}
\label{spike_table}
\end{table}

% \newpage

\noindent
{\bf Financial data. } Here, we consider $n=2000$ observations from the 
financial dataset~F.2 of \cite{our}. This consists of tick-by-tick price 
changes of the Facebook stock price, quantised to three values:
$x_i =0$ if the price goes down, $x_i=1$ if it stays the same,
and $x_i=2$ if it goes up.
The BCT entropy-rate posterior is shown in Figure~\ref{fb_h_hist}:  
It has mean $\mu = 0.921$, and standard deviation $\sigma = 0.028$.

Once again, as the ``true'' value of the entropy rate we take
the estimate produced by the CTW estimator on a 
longer sequence with $n=10^4$ observations,
giving $\bar{H}=0.916$. The results of all five estimators
are summarised in Table~\ref{table_fb},
where for the BCT estimator we once again give the posterior mean.

\begin{table}[!ht]
\centering
\begin{tabular}{ccccccccccc}
\midrule 
 & ``True" & BCT & CTW & PPM &LZ & $k=5$ & $k=6$ & $k=7$ & $k=10$  \\
\midrule
$\widehat H $ & 0.916 & \bf{0.921} & 0.939 & 1.049 & 0.846 & 0.930 & 0.907 & 0.870 & 0.713 \\
\bottomrule
\end{tabular}
\caption{Entropy rate estimates for the financial data set.}
\label{table_fb}
\end{table}

\begin{figure}[!ht]
\begin{subfigure}{0.49 \linewidth}
 \includegraphics[width= 1 \linewidth]{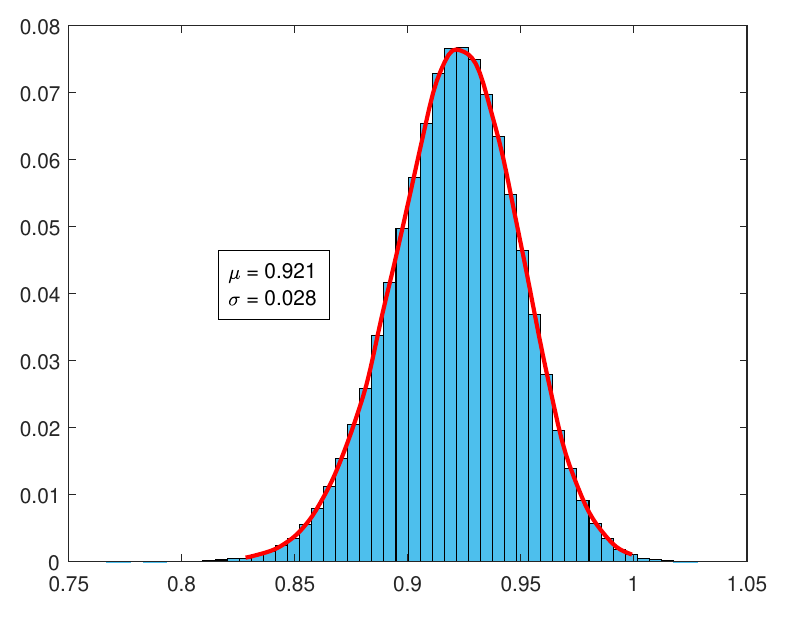}
\vspace*{-0.65 cm}
\caption{financial dataset}
\label{fb_h_hist}
\end{subfigure}
\begin{subfigure}{0.49 \linewidth}
 \includegraphics[width= 1 \linewidth]{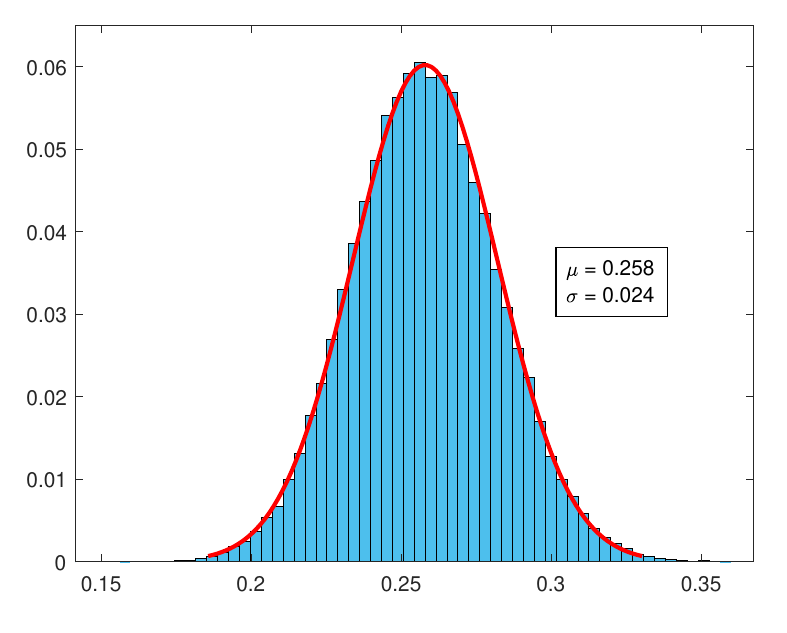}
\vspace*{-0.65 cm}
\caption{pewee birdsong}
\label{pewee_h_hist}
\end{subfigure}
\vspace*{-0.2 cm}
\caption{Histograms of the posterior 
distribution $\pi(\bar{H}|x)$ of the entropy rate,
constructed from $N=10^5$ i.i.d.\ samples in each case.}
\vspace*{-0.3 cm}
\end{figure}

\noindent
{\bf Pewee birdsong. } The last data set
examined is a time series $x$ describing
the twilight song of the wood pewee 
bird \citep{craig1943song,sarkar2016bayesian}.
It consists of $n=1327$ observations from an 
alphabet of size $m=3$. 
The BCT posterior is shown in Figure \ref{pewee_h_hist}: 
It is approximately Gaussian with mean $\mu = 0.258$ and standard 
deviation $\sigma = 0.024$. The fact that the standard deviation
is small is important as it suggests ``confidence'' in the 
resulting estimates, which is important because here
(as in most real applications) 
there is no knowledge of a ``true'' underlying value. 
Table~\ref{pewee_table} shows all the resulting estimates;
the posterior mean is shown for the BCT estimator.

\begin{table}[!ht]
\centering
\begin{tabular}{cccccccccc}
\toprule 
& BCT & CTW & PPM &LZ & $k=2$ & $k=5$ & $k=10$ & $k=15$  \\
\midrule
$\widehat H $ & 0.258 &  0.278 & 0.318 & 0.275 & 0.776 & 0.467 & 0.336 & 0.272 \\
\midrule
\end{tabular}
\caption{Entropy rate estimates for the pewee song data.}
\label{pewee_table}
\end{table}

%\medskip

\noindent
{\bf Summary. }
The main conclusion from the results on the six data
sets examined in this section is that the BCT estimator
gives the most accurate and reliable results among the
five estimators considered. In addition to the fact that
the BCT point estimates typically outperform those produced 
by other methods, the BCT estimator is accompanied by the
entire posterior distribution $\pi(\bar{H}|x)$ of the
entropy rate, induced by the observations~$x$. 
As usual, this
distribution can be used to quantify the uncertainty 
in estimating $\bar{H}$, and it contains significantly 
more information than simple point estimates and their
associated confidence intervals.

 \vspace{-0.1 cm}

\section{Concluding remarks}

In this work, we revisited the Bayesian Context Trees (BCT) 
modelling framework, which was recently found to be very effective 
for a range of statistical tasks in the analysis of discrete 
time series. 
We showed that the prior and posterior distributions on model 
space admit simple and easily interpretable representations 
in terms of branching processes, and we demonstrated
their utility both in theory and in practice.

The branching process representation was first employed to develop 
an efficient Monte Carlo sampler that provides  i.i.d.\ samples 
from the joint posterior on models and parameters,
thus facilitating effective Bayesian inference
with empirical time series data.
Then, it was used to establish strong theoretical results 
on the asymptotic consistency of the BCT posterior
on model space, 
which provide important theoretical justifications for the use
of the BCT framework in practice.
Finally, the performance of the proposed Monte Carlo sampler 
was examined extensively in the context of entropy estimation.
The resulting fully-Bayesian entropy estimator
was found to outperform several of the state-of-the-art 
approaches, on simulated and real-world~data. 

\smallskip

Although the BCT framework was originally developed 
for modelling and inference of discrete-valued time 
series, it was recently used to develop general 
mixture models for real-valued time series,
along with a collection of associated algorithmic
tools for inference \citep{bct_ar}. 
Extending the results presented in this work to 
that setting presents an interesting direction of further 
research, motivated by important practical 
applications.

\section*{Acknowledgments}

%\addcontentsline{toc}{section}{Acknowledgments}

We are grateful to Georgia Gregoriou
for providing us with the spike train data of Section~\ref{s:entropy}.

\bibliographystyle{ba}
%\bibliography{my_ref_fin}

%And this is an acknowledgements section with a heading that was produced by the
%$\backslash$section* command. Thank you all for helping me writing this
%\LaTeX\ sample file.

%\newpage

\newpage

%\section*{}
%\addcontentsline{toc}{section}{Supplementary material}

\appendix

\begin{center}
%\vspace*{-1.2 cm}

   \bf  \huge {Supplementary material}
\end{center}

\section{Proof of Proposition~3.2}

We need
the following representation of the marginal likelihood
from \citep{our}.

\begin{lem} \label{marg_lik}
The marginal likelihood $P(x|T)$ 
of the observations
$x$ given a model $T$~is,
$$P(x|T)
=\int P(x|\theta,T)\pi(\theta|T)d\theta
=\prod_{s\in T}P_{e,s},$$
where $P_{e,s}$ are the estimated probabilities in~{\em (2.4)}. 
\end{lem}

\noindent
{\bf Proof of Proposition~3.2. }
The proof parallels that of 
Proposition~3.1.
When $D=0$, $\mathcal {T}(D)$ consists of a single tree,
$T=\{\lambda\}$, which has probability~1 under 
both the BCT posterior $\pi(\cdot|x)$ and under the distribution
$\pi_b(\cdot)$ induced by the branching 
process construction. 
Suppose $D\geq 1$. 

As before,
we view every tree $T\in\mathcal {T}(D)$
as a collection of $k$ 
of $m$-branches, and we proceed 
by induction on $k$. 
For $k=0$,
i.e., for $T=\{\lambda\}$,
by the definitions,
$$
\pi_b(\{\lambda\})=P_{b,\lambda}=
\frac{\beta P_{e,\lambda}}{P_{w,\lambda}}
=\pi_D(\{\lambda\};\beta)
\frac{P_{e,\lambda}}{P_{w,\lambda}},
$$
and using
Lemma~A.1 and the fact that $P_{w,\lambda}$ is exactly the normalising constant~$P(x)$, 
$$
\pi_b(\{\lambda\})
=
\frac{
\pi_D(\{\lambda\};\beta) P(x|\{\lambda\})
}{P(x)}=\pi(\{\lambda\}|x).
$$

Now assume the result of the proposition
holds for all trees with $k$ $m$-branches, 
and suppose $T'\in\mathcal {T}(D)$ contains 
$(k+1)$ $m$-branches and is 
obtained from some $T\in\mathcal {T}(D)$
by adding a single $m$-branch
to one of its leaves, $s$.
Again, consider two~cases.

$(i)$
If $s$ is at depth $D-2$ or smaller,
then by 
construction,
\begin{align}
\pi_b(T')
=&
	\frac{\pi_b(T)}{P_{b,s}} (1-P_{b,s})\prod_{j=0}^{m-1}P_{b,sj}   , 
	\nonumber
\end{align}
and therefore, using the inductive hypothesis,
\begin{equation}
\pi_b(T') 
=
	\pi(T|x)\Big(\frac{1-P_{b,s}}{P_{b,s}}\Big)\prod_jP_{b,sj}
=
	\pi(T'|x)
	% \left[
	\frac{\pi(T|x)}{\pi(T'|x)}
	% \right]
	\Big(\frac{1-P_{b,s}}{P_{b,s}}\Big)\prod_jP_{b,sj}.
\label{eq:induct1}
\end{equation}
Using the definitions of 
$\pi_D$ and $P_{b,s}$, as well as Lemma~A.1,
we can express the 
posterior 
odds 
$\frac{\pi(T'|x)}{\pi(T|x)}$
in~(\ref{eq:induct1}) as,
\begin{align*}
\frac{\pi_D(T';\beta)}{\pi_D(T;\beta)}
	\frac{P(x|T')}{P(x|T)}
=&
	\frac{\pi_D(T';\beta)}{\pi_D(T;\beta)}
	\frac{\prod_{j=0}^{m-1}P_{e,sj}}{P_{e,s}} 
	\\
=&
	\frac{\beta^m(1-\beta)}{\beta}
	\frac{\prod_{j}P_{e,sj}}{P_{e,s}}\\
=&
	(1-P_{b,s})\frac{1-\beta}{\beta P_{e,s}(1-P_{b,s})}
	\prod_{j}\beta P_{e,sj}\\
=&
	(1-P_{b,s})
	\frac{P_{w,s}}{\beta P_{e,s}}
	\frac{1-\beta}{(P_{w,s}-\beta P_{e,s})}
	\prod_{j}\beta P_{e,sj},
\end{align*}
and 
from the definitions of $P_{w,s}$ and $P_{b,s}$
we obtain,
\begin{align}
\frac{\pi(T'|x)}{\pi(T|x)}
=
	(1-P_{b,s})
	\frac{1}{P_{b,s}}
	\frac{1}{\prod_jP_{w,sj}}
	\prod_{j}\beta P_{e,sj}
=
	\Big(\frac{1-P_{b,s}}
	{P_{b,s}}\Big)
	\prod_{j}P_{b,sj}.
\label{eq:braodds1}
\end{align}
Substituting~(\ref{eq:braodds1}) into~(\ref{eq:induct1})
yields, $\pi_b(T')=\pi(T'|x)$, as claimed.

\smallskip

\noindent $(ii)$ Similarly, if $s$ is at depth $D-1$, from the
inductive hypothesis,
\begin{align*}
\pi_b(T')
=
	\frac{\pi_b(T)}{P_{b,s}} (1-P_{b,s})
=
	\pi(T|x)\Big(\frac{1-P_{b,s}}{P_{b,s}}\Big)
=
	\pi(T'|x)
	% \left[
	\frac{\pi(T|x)}{\pi(T'|x)}
	%\right]
	\Big(\frac{1-P_{b,s}}{P_{b,s}}\Big),
%\label{eq:induct2}
\end{align*}
where the 
posterior 
odds
can be expressed as,
\begin{align} \label{odds}
\frac{\pi(T'|x)}{\pi(T|x)}
=&
	\frac{\pi_D(T';\beta)}{\pi_D(T;\beta)}
	\frac{\prod_{j=0}^{m-1}P_{e,sj}}{P_{e,s}} \nonumber
	\\
=&
	\Big(\frac{1-\beta}{\beta}\Big)
	\frac{\prod_{j}P_{w,sj}}
	{P_{e,s}} \nonumber
	\\
=&
	\frac{P_{w,s}-\beta P_{e,s}}{\beta P_{e,s}}=
	\frac{1-P_{b,s}}{P_{b,s}},
\end{align}
where in the second equality we used that $P_{w,sj} = P_{e,sj}$, as all nodes $sj$ are at depth $d=D$ in this case.
Substituting (\ref{odds}) above
yields $\pi_b(T')=\pi(T'|x)$, and completes the proof.\qed

\vspace*{-0.2 cm}

\section{Proofs of results from Section~4}

\subsection{Proof of Lemma~4.2}

Here we establish the two missing steps in the proof of the
lemma given in Section 4.1 of the main text. 

\smallskip

\noindent
{\bf Proof of~(4.5). }
Using the upper bound of Lemma~4.1 for a fixed context $s$,
and the corresponding lower bound for the context $sj$,
we obtain the upper bound,
\begin{align}
\frac{1}{M_s} \left ( \log P_{e,s} - \sum _ {j=0} ^ {m-1} \log P_{e,sj} \right ) \leq& \sum_{i=0} ^ {m-1} \frac{a_s(i)} {M_s}\log\frac{a_s(i)}{M_s} -\sum _{j =0} ^ {m-1}\frac{M_{sj}}{M_s} \sum_{i=0}^{m-1}\frac{a_{sj}(i)} {M_{sj}}\log\frac{a_{sj}(i)}{M_{sj}} \nonumber \\ 
+ & \frac{m-1}{2 M_s} \left ( \sum _ {j=0} ^ {m-1} \log M_{sj}   -\log M_s \right ) + \frac {C}{M_s}, \label{expansion1}
\end{align}
for some constant $C$. Since $M_s$ and $M_{sj}$ both tend to infinity
a.s.\ as $n\to\infty$ by positive-ergodicity, the last two terms
above both vanish a.s.

For the first two terms, we first note that,
by the ergodic theorem for Markov chains 
(e.g., \cite[p.~92]{chung1967markov}),
\begin{align}
\frac {a_s (i)} {M_s} = \frac {a_s (i)} {n}  \  \frac {n} {M_s} \to \frac{\pi(si)}{\pi (s)} = \pi(i|s), \quad \text{a.s.},
\end{align}
where for the stationary distribution $\pi$, the notation we use is that $si$ denotes the concatenation of context $s$ followed by symbol $i$
moving `forward' in time.

\newpage

\noindent Recalling the definition of $X$ and $J$,
we have,
\begin{equation}
\mathbb P (X =i | s) = \frac{\pi(si)}{\pi (s)} =  \pi(i|s), \quad \mathbb P (J =j | s) =  \frac{\pi (js)}{\pi (s)}, 
\end{equation}
so that for the first term of (\ref{expansion1}), as $n\to\infty$,
\begin{equation}\label{entropy1}
 \sum_{i=0} ^ {m-1} \frac{a_s(i)} {M_s}\log\frac{a_s(i)}{M_s} \to - H(X|s), \quad \text{a.s.}
\end{equation}
Similarly, for the second term of (\ref{expansion1}), from the ergodic theorem,
\begin{align}
& \frac {a_{sj} (i)} {M_{sj}} =\frac {a_{sj}(i)} {n}  \  \frac {n} {M_{sj}} \to \frac{\pi(jsi)}{\pi (js)} = \pi(i|js) =\mathbb P (X =i | s , J =j)  , \quad \text{a.s.}, \\
& \frac{M_{sj}}{M_s} = \frac{M_{sj}}{n}  \ \frac{n}{M_s}  \to \frac{\pi (js)}{\pi (s)} =  \mathbb P (J =j | s) ,  \quad \text{a.s.}, 
\end{align}
so that, as $n\to\infty$,
\begin{align}
- \sum _{j =0} ^ {m-1}\frac{M_{sj}}{M_s} \sum_{i=0}^{m-1}\frac{a_{sj}(i)} {M_{sj}}\log\frac{a_{sj}(i)}{M_{sj}} \to \sum _ {j=0} ^ {m-1}  \mathbb P (J =j | s)   H (X |s, J=j) = H(X|s,J), 
\end{align}
by the definition of conditional entropy. 
Finally, combining with (\ref{entropy1}), we get,
\begin{align}
 \sum_{i=0} ^ {m-1} \frac{a_s(i)} {M_s}\log\frac{a_s(i)}{M_s}- \sum _{j =0} ^ {m-1}\frac{M_{sj}}{M_s} \sum_{i=0}^{m-1}\frac{a_{sj}(i)} {M_{sj}}\log\frac{a_{sj}(i)}{M_{sj}} \to - I(X;J|s), \quad \text{a.s.}
\end{align}

Following the same sequence of steps, we can obtain a
lower bound corresponding to~(\ref{expansion1}) as,
\begin{align}
\frac{1}{M_s} \left ( \log P_{e,s} - \sum _ {j=0} ^ {m-1} \log P_{e,sj} \right ) \geq& \sum_{i=0} ^ {m-1} \frac{a_s(i)} {M_s}\log\frac{a_s(i)}{M_s} -\sum _{j =0} ^ {m-1}\frac{M_{sj}}{M_s} \sum_{i=0}^{m-1}\frac{a_{sj}(i)} {M_{sj}}\log\frac{a_{sj}(i)}{M_{sj}} \nonumber \\ 
+ & \frac{m-1}{2 M_s} \left ( \sum _ {j=0} ^ {m-1} \log M_{sj}   -\log M_s \right ) + \frac {C'}{M_s}, \label{expansion2}
\end{align}
where the only difference from (\ref{expansion1}) is the constant $C'$. Therefore,
\begin{equation*}
\frac{1}{M_s} \left ( \log P_{e,s} - \sum _ {j=0} ^ {m-1} \log P_{e,sj} \right ) \to - I(X;J|s), \quad \text{a.s.},
\end{equation*}
or, equivalently,
\begin{equation}
\log P_{e,s} - \sum _ {j=0} ^ {m-1} \log P_{e,sj}  = - M_s I(X;J|s) + o(M_s), \quad \text{a.s.}
\end{equation}
And since $M_s = n \pi(s) + o(1)$ a.s.\ by the ergodic
theorem, we get~(4.5).
\qed

\noindent
{\bf Proof of final step in Lemma~4.2. }
As already noted,~(4.5) implies $P_{b,s} \to 0$ a.s.\
for nodes whose children are leaves of $T^*$. 
The same holds for all internal nodes $s$ of $T^*$ 
for which $I(X;J|s)>0$.
The only remaining case is that
of internal nodes $u$ for which $I(X;J|u)=0$.
The fact that again
$P_{b,u} \to 0$ a.s.\ is an immediate
consequence of the result given as
Lemma~\ref{ind_inner} 
in Section~\ref{rates}.
\qed

\subsection{Proof of Lemma~4.3}

Here we provide proofs for the two missing steps in the proof
of the lemma given in Section~4.1 of the main text.

\noindent
{\bf Proof of~(4.6). }
We can rewrite (\ref{expansion2}) as, 
\begin{align}
\sum _ {j=0} ^ {m-1} \log P_{e,sj}  - \log P_{e,s} &\leq \sum _{j =0} ^ {m-1}\sum_{i=0}^{m-1}{a_{sj}(i)} \log\frac{a_{sj}(i)}{M_{sj}} -  \sum_{i=0} ^ {m-1} {a_s(i)} \log\frac{a_s(i)}{M_s} \nonumber \\
&- \frac{m-1}{2} \left ( \sum _ {j=0} ^ {m-1} \log M_{sj}   -\log M_s \right ) - {C'}. \label{expansion3}
\end{align}
We write $\widehat p _ s $ and $\pi _ s$ for the empirical and 
stationary conditional distributions of the symbol following $s$, 
\[
\widehat p _ s (i) := \frac{a_s(i)} {M_s}, \quad \pi _ s (i) := \pi(i|s) = \frac{\pi(si)}{\pi (s)},\qquad i\in A.
\]
Let $D(p\|q)$ denote the relative entropy 
(or Kullback-Leibler divergence) between 
two probability mass functions $p,q$ on the same
discrete alphabet
\citep[Ch.~2]{cover1999elements}. By the nonnegativity of
relative entropy we have,
\begin{eqnarray*}
\sum _{i=0} ^ {m-1}  a _ s (i)  \log \pi _ s (i) -  
\sum_{i=0} ^ {m-1} {a_s(i)} \log\frac{a_s(i)}{M_s}
 = 
	M_s  \sum_{i=0} ^ {m-1} {\widehat p _ s (i) } 
	\log \frac{\pi _ s (i)}{\widehat p _ s (i) } 
=  - M_s D(\widehat p _ s \| \pi _ s) \leq 0.
\end{eqnarray*}
Adding and subtracting the 
term $\sum _{i=0} ^ {m-1}  a_s(i)\log \pi_s (i)$ 
to (\ref{expansion3}) and using the last inequality,
\begin{align}
\sum _ {j=0} ^ {m-1} \log P_{e,sj}  - \log P_{e,s} &\leq \sum _{j =0} ^ {m-1} \sum_{i=0}^{m-1} a _ {sj} (i)  \log\widehat p _ {sj} (i)  -   \sum_{i=0} ^ {m-1} a_ {s} (i) \log  \pi _ s (i)  \nonumber \\
&- \frac{m-1}{2} \left ( \sum _ {j=0} ^ {m-1} \log M_{sj}   -\log M_s \right ) - {C'}. \label{expansion4}
\end{align}
We examine the first and second terms in~(\ref{expansion4})
separately. For the first (and main) term,
since
for all count vectors, $a_s(i) = \sum _ {j=0} ^ {m-1} a_{sj} (i)$, 
we can express,
\begin{align*}
\sum _{j =0} ^ {m-1} \sum_{i=0}^{m-1} a _ {sj} (i)  
\log\widehat p _ {sj} (i)   -   \sum_{i=0} ^ {m-1} 
	a_ {s} (i) \log  \pi _ s (i) 
= & 
	\sum _{j =0} ^ {m-1} \sum_{i=0}^{m-1} a _ {sj} (i)  
	\log \frac{\widehat p _ {sj} (i) } { \pi _ s (i)  }\\
= & 
	\sum _{j =0} ^ {m-1} M_{sj} \sum_{i=0}^{m-1}  \widehat p _ {sj} (i)
	\log \frac{\widehat p _ {sj} (i) } { \pi _ s (i)  } \\
= &  
	\sum _{j =0} ^ {m-1} M_{sj} D ( \widehat p _ {sj} \| \pi _ s )\\
= &  
	\sum _{j =0} ^ {m-1} M_{sj} D ( \widehat p _ {sj} \| \pi _ {sj} ) ,
\end{align*} 
where the last equality holds because $s$ is either
a leaf or an external nodes of $T ^*$, so that
$I(X;J|s)=0$ and $\pi_{sj} = \pi _ {sj'} = \pi_s$, for all $j,j'$. 

\newpage

In order to bound the 
relative entropy between the empirical and the stationary conditional 
distributions, we first recall that 
relative entropy is bounded above by the $\chi ^ 2$-distance,
e.g., \citep{gibbs:02},
\begin{equation}\label{chi_bound}
 D ( \widehat p _ {sj} \| \pi _ {sj} ) \leq d _ { \chi ^2} ( \widehat p _ {sj} , \pi _ {sj} )  = \sum _ {i=0} ^ {m-1} \frac{( \widehat p _ {sj} (i) - \pi _ {sj}(i) ) ^ 2}{\pi _ {sj}(i) }.
\end{equation}
From the law of the iterated logarithm (LIL) for 
Markov chains \citep[p.~106]{chung1967markov}, we have,
a.s.\ as $n\to\infty$,
\begin{align} \label{lil}
a_{sj} ( i) = n \pi(jsi) + O (\sqrt {n \log \log n}) , \quad M_{sj} = n \pi (js)+ O (\sqrt {n \log \log n}),
\end{align}
so that, 
$$\widehat p _ {sj}(i) 
= \frac{a_{sj} ( i)}{M_{sj}} 
= \frac{\pi(jsi)}{\pi (js)}
+ O\left ( \sqrt { \frac{\log \log n}{n}}\right )
= \pi _{sj} (i) 
+ O\left ( \sqrt { \frac{\log \log n}{n}}\right ).
$$
Substituting in (\ref{chi_bound}) yields,
\begin{equation}\label{rel_entr_bound_fin}
 D ( \widehat p _ {sj} \| \pi _ {sj} ) \leq  O\left ( { \frac{\log \log n}{n}}\right ) \sum _ {i=0} ^ {m-1}  \frac{1}{\pi _ {sj}(i) } = O\left ( { \frac{\log \log n}{n}}\right ),  \quad \text{a.s.},
\end{equation}
and finally, using (\ref{lil}) again,
\begin{equation}\label{ologlogn}
  \sum _{j =0} ^ {m-1} M_{sj} D ( \widehat p _ {sj} \| \pi _ {sj} ) =O (\log \log n) , \quad \text{a.s.}
\end{equation} 

\noindent For the third term in~(\ref{expansion4}), 
$$
 \sum _ {j=0} ^ {m-1} \log M_{sj}   -\log M_s  = \sum _ {j=0} ^ {m-1} \log \frac {M_{sj}} {M_s} + (m-1) \log M_s .
$$
Using the LIL, 
$ M_{s} = n \pi (s)+ O (\sqrt {n \log \log n})$, a.s., so,
$$\log M_{s}  
= \log \left ( n \pi (s) 
\left \{ 1 +  O (\sqrt {\log \log n / n }) \right \} \right )  
= \log n + O (1), \quad \text{a.s.}
$$
And, using LIL again as above, 
$M_{sj} / M_s = \pi(js) / \pi(s) + O (\sqrt {\log \log n / n}) = O (1)$, 
a.s., so that,
$ \log  {M_{sj}}/ {M_s}= O(1)$, a.s.,
and,
\begin{equation}
  \sum _ {j=0} ^ {m-1} \log M_{sj}   -\log M_s    = (m-1) \log n + O (1) , \quad \text{a.s.},
\end{equation} 
which together with (\ref{expansion4}) and (\ref{ologlogn}) complete the 
proof of equation (4.6). \qed

\smallskip

\noindent
{\bf Proof of final step in Lemma~4.3. }
As discussed in the proof of Lemma~4.3 in the main text, 
the asymptotic relation~(4.6) implies that  
$\prod _j P_{e,sj}/ {P_{e,s}} \to 0$, a.s., for all leaves and 
external nodes $s$ of $T^*$. The next proposition states that 
it also implies that  $\prod _j P_{w,sj}/ {P_{e,s}} \to 0$, 
so that by~(4.3) $P_{b,s} \to 1$, a.s., completing the proof of Lemma~4.3.
Note that it suffices to consider leaves at depths $d\leq D-1$,
since for leaves $s$ at depth~$D$ we already have $P_{b,s}=1$.
\qed

\begin{prop}
\label{ind_leaves}
Under the assumptions of Theorem~4.1,
for all leaves and external nodes $s$ of $T^*$ at
depths $d\leq D-1$ we have,
as $n\to\infty$:
\begin{equation*} 
\frac { \prod _ {j=0} ^ {m-1} P_{w,sj}} { {P_{e,s}}} \to 0, \quad \text{a.s.}
\end{equation*}
\end{prop}

\noindent
{\bf Proof. }
Note that, since the stationary distribution is positive on all
finite contexts, the tree $T_{\text{MAX}}$ is eventually a.s.\ the complete 
tree of depth $D$, so we need not consider special cases of contexts
$s$ that do not appear in the data separately.
Let $s$ be a leaf or external node of $T^*$ at depth 
$0 \leq d \leq D-1$.  The proof is by induction on $d$. 

\smallskip

For $d=D-1$, the claim is satisfied trivially as $P_{w,sj} = P_{e,sj}$, since nodes $sj$ are at depth~$D$. For the inductive step, we assume that the claim holds for all leaves and external nodes $s$ of $T^*$ at some depth $d\leq D-1$, and 
consider a leaf or external node $s$ of $T^*$ at depth $d-1$. We have,
as $n\to\infty$,
\begin{align}
\frac { \prod _ {j=0} ^ {m-1} P_{w,sj}} { {P_{e,s}}} 
=& \frac { \prod _ {j=0} ^ {m-1}\big [\beta P_{e,sj} +(1 - \beta ) \prod _ {t=0} ^ {m-1} P_{w,sjt}\big]} { {P_{e,s}}} \nonumber \\
=& \frac { \prod _ {j=0} ^ {m-1} \big[ \beta P_{e,sj} +(1 - \beta ) \prod _ {t=0} ^ {m-1} P_{w,sjt}\big]} { \prod _ {j=0} ^ {m-1} \beta P_{e,sj}} \ \frac  { \prod _ {j=0} ^ {m-1} \beta P_{e,sj}} { {P_{e,s}}}  \nonumber \\
= & \prod _ {j=0} ^ {m-1} \left ( 1 +\frac {1-\beta} { \beta} \frac { \prod _ {t=0} ^ {m-1} P_{w,sjt}}{P_{e,sj}}  \right )  \ \frac  { \prod _ {j=0} ^ {m-1}  P_{e,sj}} { {P_{e,s}}} \beta ^ m \to 0 , \quad \text {a.s.}, \label{eq_ind_leaves}
\end{align}
as $  { \prod _ {t=0} ^ {m-1} P_{w,sjt}} / {P_{e,sj}} \to 0$ by the inductive hypothesis, and $  { \prod _ {j=0} ^ {m-1}  P_{e,sj}} /{ {P_{e,s}}} \to 0$ for 
a node $s$ which is either a leaf or external node of $T^*$.
This establishes the inductive step and completes the proof.
\qed

\subsection{Proof of Theorem~4.2} \label{rates}

The starting point of the proof is the representation 
of the posterior given in equation~(3.2) of the main
text. In particular, we examine the asymptotic
behaviour of the branching probabilities
$P_{b,s}$ separately for leaves and internal~nodes.

\noindent
{\bf Leaves. } 
Let $s$ be a leaf or an external node of $T^*$.
We already have a strong upper bound for the
estimated probabilities of $s$ in~(4.6).
Write $r = (m-1) ^ 2 /2 $.
Since the bound~(4.6) holds a.s., a straightforward
sample-path-wise computation immediately 
implies that, for all $\epsilon>0$,
\begin{align}
\frac {\prod _ {j=0} ^ {m-1} P _ {e,sj}}{P_{e,s}} 
=O\big(n^{-r+ \epsilon}\big), \quad \text {a.s.}
\label{eq:Pesj}
\end{align}
Proposition~\ref{ind_rate} states that 
$\prod _j P_{w,sj}/ {P_{e,s}}$
has the same asymptotic behaviour.

\begin{prop}\label{ind_rate}
For all leaves and external nodes $s$ of $T^*$
at depths $d\leq D-1$, for any $\epsilon>0$ we have
as $n\to\infty$:
\begin{equation*}
\frac { \prod _ {j=0} ^ {m-1} P_{w,sj}} { {P_{e,s}}}  = 
{O} \left (  n^{-\frac{(m-1)^2}{2} + \epsilon  } \right ), \quad \text{a.s.}
\end{equation*}
\end{prop}

\noindent 
{\bf Proof. } 
The proof is similar to that of Proposition~\ref{ind_leaves},
by induction on $d$. 

\smallskip

If $d=D-1$, 
then $P_{w,sj}=P_{e,sj}$ and the claim 
follows from~(\ref{eq:Pesj}).
For the inductive step, assume the claim holds for 
all leaves and external nodes at some depth $d\leq D-1$,
and consider a leaf or external node $s$ at depth $d-1$.
Then substituting~(\ref{eq:Pesj})
into~(\ref{eq_ind_leaves}) and noting
that ${ \prod _ t P_{w,sjt}} / {P_{e,sj}} \to 0$ a.s.\
by Proposition~\ref{ind_leaves}, completes the proof. 
\qed

Combining Proposition~\ref{ind_rate} with
equation~(4.3), we get that,
for any leaf or external node $s$ at depth $d$,
a.s.\ as $n\to\infty$:
\begin{equation}
P_{b,s} = 1 - {O} \left (  n^{-\frac{(m-1)^2}{2}+\epsilon}\right ),
\quad\mbox{if }d\leq D-1,
\qquad
\mbox{and }
P_{b,s}=1,
\quad\mbox{if }d=D.
\label{eq:boundleaf}
\end{equation}

\noindent
{\bf Internal nodes. } As in the proof of Lemma~4.2, 
we first consider internal nodes 
whose children are leaves of $T^*$, so that $I=I(X;J|s) > 0$. 
For these nodes, equation (4.5) gives,
\begin{align}
 \frac {P_{e,s}}{\prod _ {j=0} ^ {m-1} P _ {e,sj}} = \exp 
\big \{- n I \pi(s)      + o(n) \big\}, \quad \text{a.s.}, \label{b32}
\end{align}
so for any $\epsilon >0 $,
\begin{align}
 \frac {P_{e,s}}{\prod _ {j=0} ^ {m-1} P _ {e,sj}}= o\big ( \exp \big  \{ -n  (1- \epsilon) I \pi (s) \big  \} \big ), \quad \text{a.s.},
\end{align}
and substituting in equation~(4.4) in the main
text we obtain the same bound for the branching
probabilities,
\begin{align}
P_{b,s} 
=  o\Big ( \exp \big  \{ -n  (1- \epsilon) I \pi (s) \big  \} \Big ), \quad \text{a.s.}  \label{p_b_internal}
\end{align}

Next we establish a corresponding bound
for all internal nodes,
indeed, for any node $u$ that is a suffix of a node $s$ that
has $I(X;J|s) > 0$.

\begin{lem}\label{ind_inner}
Let $s$ be a context of length $l(s) \leq D-1$ for which $I(X;J|s) > 0$, and let~$u$ be any suffix of~$s$. Then, for any $\epsilon > 0 $,
we have as $n\to\infty$:
\begin{equation}
P_{b,u} = o\Big ( \exp \big \{ -n  (1- \epsilon) I(X;J|s)  \pi (s) \big \}\Big ), 
\quad \text{a.s.}
\label{eq:lemmaO}
\end{equation}
\end{lem}

\noindent
{\bf Proof. }
Let $\Delta l = l(s) -l(u) \geq 0$; the proof is by induction on $\Delta l$.
For $\Delta l=0 $, the claim is satisfied trivially as $u=s$, for which 
$I(X;J|s) > 0$, corresponding to the previous case.

\smallskip

For the inductive step, we assume that the claim holds for context $uj$ 
which is the suffix of $s$ with $\Delta l =k \geq 0$, and prove that 
it also holds for the context $u$, which is the suffix of $s$ 
with $\Delta l =k+1 $. For node $u$, which is at depth $d<D-1$, 
from~(4.4) and the definition of the branching probabilities,
$$
P_{b,u} \leq \Big(\frac{\beta}{1-\beta}\Big)
\frac{P_{e,u}}{\prod _ {t=0} ^ {m-1}  P_{w,ut}}
% = \Big(\frac{\beta}{1-\beta}\Big)  \frac{P_{e,u}}{\prod _ {t=0} ^ {m-1}  P_{e,ut}}  
% \frac {\prod _{t=0} ^ {m-1}  P_{e,ut}} {\prod _ {t=0} ^ {m-1}  P_{w,ut}} 
=  C  \frac{P_{e,u}}{\prod _ {t=0} ^ {m-1}  P_{e,ut}}  \prod _{t=0} ^ {m-1}  P_{b,ut},
$$
where the constant $C=[\beta^{m-1}(1-\beta)]^{-1}$.
And, as $P_{b,ut}\leq 1 $ for all $t$, we can further bound,
\begin{align} \label {b37}
P_{b,u}  \leq C \   \frac{P_{e,u}}{\prod _ {t=0} ^ {m-1}  P_{e,ut}} \   P_{b,uj},
\end{align}
keeping only the specific child $uj$ of $u$ 
which is a suffix of $s$. 

\newpage

\noindent From (\ref{b32}) for node $u$, 
we know that, even if $I(X;J|u)$ is zero, we have,
\begin{equation} 
 \frac {P_{e,u}}{\prod _ {t=0} ^ {m-1} P _ {e,ut}} =
\exp( o(n)),  \quad \text{a.s.},
\end{equation}
and combining this with~(\ref{b37}) and the inductive hypothesis
that~(\ref{eq:lemmaO}) holds for $uj$ in place of $u$,
$$
P_{b,u} = o\Big ( \exp \big \{ -n  (1- 2\epsilon ) I(X;J|s)  \pi (s) \big \}\Big ), 
\quad \text{a.s.},
$$
completing the proof of the inductive step and the proof of the lemma.
\qed

Substituting the bounds
on the branching probabilities on the leaves~(\ref{eq:boundleaf}) 
and on the internal nodes~(\ref{eq:lemmaO}) into the expression
for the posterior of $T^*$ in equation~(3.2) in the main text,
yields the result claimed in Theorem~4.2.
Finally, a simple examination of~(\ref{eq:boundleaf}) 
and~(\ref{eq:lemmaO}) in the case when $T^*$ is the full
tree of depth $D$ shows that $P_{b,s}=1$ for all leaves $s$,
so the rate is determined by the exponential bounds
in~(\ref{eq:lemmaO}) as claimed in Corollary~4.1.

\vspace*{-0.3 cm}

%\section{Additional experiments}
\section{Entropy estimation} \label{D}

\vspace*{-0.2 cm}

This section contains additional details associated with the entropy 
estimation experiments of Section 5.2 of the main text.

\medskip

\noindent
{\bf A ternary chain. } The 
parameters of this 
chain are:
\begin{align*}
&
\hspace{-0.18in} 
	\theta_1=(0.4,0.4,0.2),\; \theta_2=(0.2,0.4,0.4),\\
& 
\hspace{-0.18in} 
	\theta_{00}=(0.4,0.2,0.4),\; \theta_{01}=(0.3,0.6,0.1),\\
&
\hspace{-0.18in} 
	\theta_{022}=(0.5,0.3,0.2),\\
& 
\hspace{-0.18in} 
	\theta_{0212}=(0.1,0.3,0.6),\;
	\theta_{0211}=(0.05,0.25,0.7),\;  \theta_{0210}=(0.35,0.55,0.1),\\
&
\hspace{-0.18in} 
	\theta_{0202}=(0.1,0.2,0.7),\;
	\theta_{0201}=(0.8,0.05,0.15),\\
& 
\hspace{-0.18in} 
	\theta_{02002}=(0.7,0.2,0.1),\; \theta_{02001}=(0.1,0.1,0.8),\;
	\theta_{02000}=(0.3,0.45,0.25).
\end{align*}
\noindent
{\bf A third order binary chain. }

\begin{figure}[!ht]
\vspace*{-0.3 cm}
\begin{subfigure}{0.48 \linewidth}
\vspace*{0.3 cm}
 \includegraphics[width= 1 \linewidth]{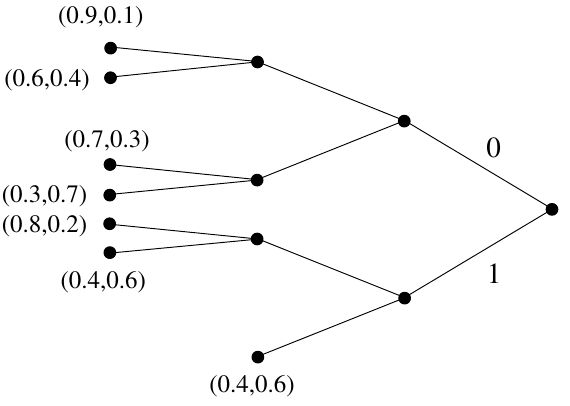}
\vspace*{-0.25 cm}
\caption{tree model and parameters}
\label{binary_tree_example}
\end{subfigure}
~
\begin{subfigure}{0.49 \linewidth}
 \includegraphics[width= 1 \linewidth]{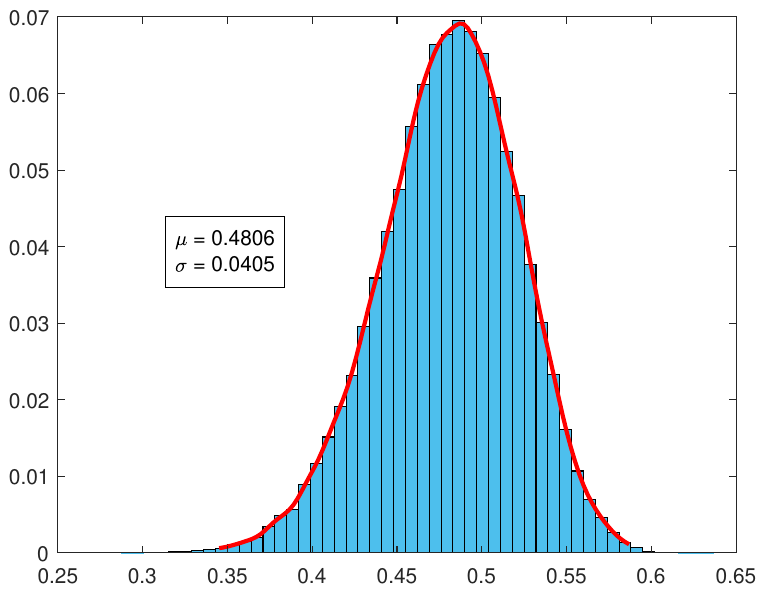}
\vspace*{-0.58 cm}
\caption{entropy rate posterior}
\label{binary_h_hist}
\end{subfigure}
\vspace*{-0.2 cm}
\caption{$(a)$~Model and parameters for the third order binary chain.
$(b)$~Histogram of the posterior $\pi(\bar{H}|x)$ of the entropy
rate given $n=1000$ observations from the chain, constructed from $N=10^5$
Monte Carlo samples.}
\end{figure}

\newpage

\noindent
The tree model for the third order binary chain, along with its associated parameters, 
is shown in Figure \ref{binary_tree_example}.
The posterior $\pi(\bar{H}|x)$ of the entropy
rate based on $n=1000$ observations $x$ is shown in Figure~\ref{binary_h_hist}: 
It is approximately Gaussian with mean $\mu = 0.4806$ and standard deviation 
$\sigma = 0.0405$. The histogram was constructed using $N=10 ^5 $ i.i.d.\ samples 
from the entropy rate posterior.

\medskip 

\noindent
{\bf A bimodal posterior. } The distribution 
of the third-order chain $\{X_n\}$ in this
example is given by,
$$\Pr(X_n=j|X_{n-1}=a,X_{n-2}=b,X_{n-3}=i)=Q_{ij},\quad
i,j,a,b\in A,$$ 
where the alphabet $A = \{0,1,2,3,4,5\}$ and 
the transition matrix $Q=(Q_{ij})$ is,
\[ Q=
\left( \begin{array}{cccccc}
0.5 &  0.2 &  0.1  &  0    &  0.05 &   0.15\\
0.4 &  0   &  0.4  &  0.2  &  0    &   0\\
0.3 &  0.1 &  0.23 &  0.12 &  0.05 &   0.2\\
0.05&  0.1 &  0.05 &  0.05 &  0.03 &   0.72\\
0   &  0   &  1    &  0    &  0    &   0\\
0.1 &  0.2 &  0.3  &  0.2  &  0.05 &   0.15
\end{array} \right).
\]
Viewed as a variable-memory chain, 
the model of $\{X_n\}$ is the complete tree of depth 3,
but the dependence of each $X_n$ on its past is only
via $X_{n-3}$. So, meaningful dependence is
detected only at memory
lengths of at least three: the two most recent symbols are independent of~$X_n$.
This is why the MAP model $T_1^*$
identified by the BCT algorithm of \cite{our}
based on the $n=1450$ observations $x$
is the 
empty tree $T_1 ^* = \{ \lambda \}$, 
corresponding to i.i.d.\ data. Its posterior probability 
is $\pi (T _ 1 ^ * | x) = 0.09 $, which corresponds exactly to the weight of the secondary mode in the entropy rate posterior, as shown in Figure 7a of the main text. All other trees identified by the $k$-BCT algorithm are 
complex trees of depth~3, with posterior probabilities
close to that of $T_ 1 ^*$; e.g., 
the second {\em a posteriori} most likely model $T_2^*$ has posterior
$\pi ( T _ 2 ^ * | x) = 0.08$. These trees of depth~3 form the other 
mode of the bimodal posterior on model space, which corresponds to 
the primary mode of the entropy rate posterior in Figure 7a. 

\smallskip

Table \ref{bimodal_h_table} shows the entropy
rate estimates by all five methods in this example.
The MAP value of the posterior $\pi(\bar{H}|x)$ is given
for BCT. Note that, although the plug-in estimator with
block-length $k=5$ gives a slightly better estimate
than the BCT here, given the well-known high variability 
of the plug-in this is likely more a coincidence rather than 
an indication of accuracy of the plug-in.

\begin{table}[!ht]
\centering
\begin{tabular}{ccccccccccc}
\midrule 
 & True & BCT & CTW & PPM &LZ & $k=3$ & $k=4$ & $k=5$ & $k=6$ \\
\midrule
$\widehat H $ & 1.355 & 1.406 & 1.643 & 1.650 & 1.283 & 1.609 & 1.481 & 1.333 & 1.173 \\
\midrule
\end{tabular}
\caption{Entropy rate estimates for the `bimodal posterior' example.}
\label{bimodal_h_table}
\end{table}

\end{document}